\documentclass[12pt]{amsart} 

\usepackage[utf8]{inputenc}
\usepackage[T1]{fontenc}	
\usepackage[french,english]{babel}

\usepackage{lmodern}			
\usepackage{newtxtext}

\usepackage[top=3cm, bottom=3cm, left=3.6cm, right=3.6cm]{geometry}

\usepackage[plainpages=false,colorlinks,linkcolor=bleuFonce,citecolor=rougeFonce,urlcolor=vertFonce,breaklinks]{hyperref}

\usepackage{graphicx}	

\usepackage{tikz}
\usetikzlibrary{patterns}

\usepackage{color}
\definecolor{vertFonce}{rgb}{0,0.5,0}
\definecolor{numLignes}{rgb}{0.17,0.57,0.7}	
\definecolor{gris}{rgb}{0.5,0.5,0.5}
\definecolor{grisFonce}{rgb}{0.2,0.2,0.2}
\definecolor{orange}{rgb}{1,0.65,0.31}		
\definecolor{orangeFonce}{rgb}{1,0.4,0}
\definecolor{bleuFonce}{rgb}{0,0,0.4}
\definecolor{rougeFonce}{rgb}{0.3,0,0}
\definecolor{rougeWord}{rgb}{0.5,0,0}
\definecolor{vertClair}{rgb}{0.8,1,0.8}
\definecolor{rougeClair}{rgb}{1,0.5,0.5}

\usepackage{pict2e}
\usepackage{multido}

\usepackage{amsfonts,amssymb,amsthm,amsmath} 
\usepackage{amsaddr}	
\usepackage{dsfont}		
\usepackage{mathrsfs}
\usepackage{yfonts}

\newtheorem{thm}{Theorem}
\newtheorem{cor}{Corollary}[section]
\newtheorem{prop}{Proposition}[section]

\newtheorem{remark}{Remark}[section]

\newenvironment{system}{%
	\equation\left\{\ \begin{aligned}%
}{%
	\end{aligned} \right. \endequation%
}
\newenvironment{system*}{%
	\equation\nonumber\left\{\ \begin{aligned}
}{%
	\end{aligned} \right. \endequation%
}

\newcommand{\step}[1]	{\paragraph{\itshape #1.}}

\usepackage{stmaryrd}
\newcommand		{\subsetArrow}	{\mathrel{\ooalign{$\subset$\cr%
\hidewidth\raise-.087ex\hbox{$_\shortrightarrow\mkern-1.5mu$}\cr}}}
\newcommand		{\subsetarrow}	{\mathrel{\ooalign{$\subset$\cr%
\hidewidth\raise-1.45ex\hbox{$\vec{}\mkern6mu$}\cr}}}

\newcommand		{\N}		{\mathbb N}			
\newcommand		{\ZZ}		{\mathbb Z}			
\newcommand		{\Z}		{\ZZ}
\newcommand		{\RR}		{\mathbb R}			
\newcommand		{\R}		{\RR}
\newcommand		{\Rd}		{\R^d}

\newcommand		{\Rdd}		{\R^{2d}}
\newcommand		{\CC}		{\mathbb C}			

\newcommand		{\cH}		{\mathcal H}		
\newcommand		{\cD}		{\mathcal D}		
\renewcommand	{\L}		{\mathcal L}		

\newcommand     {\cW}		{\mathcal W}		
\newcommand		{\cB}		{\mathcal B}

\newcommand		\sfH		{\mathsf H}			
\newcommand		\sfT		{\mathsf T}			

\newcommand		{\lt}			{\left}				%
\newcommand		{\rt}			{\right}			%
\renewcommand	{\(}			{\lt(}
\renewcommand	{\)}			{\rt)}
\newcommand		{\set}[1]		{\lt\{#1\rt\}}

\newcommand		{\com}[1]		{\lt[{#1}\rt]}		

\newcommand		{\n}[1]			{\lt|{#1}\rt|}
\newcommand		{\nrm}[1]		{\lt\|{#1}\rt\|}
\newcommand		{\snrm}[1]		{\lVert #1\rVert}
\newcommand		{\bnrm}[1]		{\big\lVert #1\big\rVert}
\newcommand		{\Bnrm}[1]		{\Big\lVert #1\Big\rVert}
\newcommand		{\Nrm}[2]		{\nrm{#1}_{#2}}
\newcommand		{\sNrm}[2]		{\snrm{#1}_{#2}}
\newcommand		{\bNrm}[2]		{\bnrm{#1}_{#2}}
\newcommand		{\BNrm}[2]		{\Bnrm{#1}_{#2}}

\newcommand		{\indic}	{\mathds{1}}		

\renewcommand		{\d}		{\mathop{}\!\mathrm{d}}		

\newcommand			{\Dx}		{\nabla_x}
\newcommand			{\Dv}		{\nabla_\xi}
\newcommand			{\conj}[1]	{\overline{#1}}		

\DeclareMathOperator{\re}		{Re}				
\DeclareMathOperator{\tr}		{Tr}				

\renewcommand	{\Re}[1]		{\re\!\( #1 \)}		
\newcommand		{\Tr}[1]		{\tr\!\( #1 \)} 	

\newcommand		{\intd}			{\int_{\R^d}}
\newcommand		{\intdd}		{\int_{\R^{2d}}}
\newcommand		{\iintd}		{\iint_{\R^{2d}}}

\newcommand		{\loc}			{\mathrm{loc}}

\newcommand		{\eps}			{\varepsilon}

\newcommand		{\cC}			{\mathcal{C}}

\usepackage{braket}
\newcommand		{\Inprod}[2]	{\Braket{#1 | #2}}

\newcommand		{\op}		{\boldsymbol{\rho}}	

\newcommand		{\opnu}		{\boldsymbol{\nu}}	

\newcommand		{\tildop}		{\,\tilde{\!\op}}	
\newcommand		{\ttildop}		{\,\tilde{\tilde{\!\op}}}	

\newcommand		{\opp}		{\boldsymbol{p}}

\newcommand		{\Dh}		{\boldsymbol{\nabla}}	
\newcommand		{\Dhx}[1]	{\Dh_{\!x} #1}			
\newcommand		{\Dhv}[1]	{\Dh_{\!\xi} #1}		

\DeclareMathOperator{\dG}	{\d\Gamma}			

\newcommand		{\CS}	{\cC^\mathrm{S}}

\title[\textsc{Regularity of Projection Operators}]{\Large Optimal Semiclassical Regularity of Projection Operators and Strong Weyl Law}
\author[\textsc{L.~Lafleche}]{\large\textsc{Laurent Lafleche}}
\address{Institut Camille Jordan, UMR 5208 CNRS \\\& Université Claude Bernard Lyon 1, France}
\curraddr{\textsc{Unit\'e de Math\'ematiques pures et appliqu\'ees, \'Ecole Normale Supérieure de Lyon, allée d’Italie, 69364 Lyon, France}}
\email{laurent.lafleche@ens-lyon.fr}

\begin{document}

\begin{abstract}\small\vspace{10pt}
	Projection operators arise naturally as one-particle density operators associated to Slater determinants in fields such as quantum mechanics and the study of determinantal processes. In the context of the semiclassical approximation of quantum mechanics, projection operators can be seen as the analogue of characteristic functions of subsets of the phase space, which are discontinuous functions. We prove that projection operators indeed converge to characteristic functions of the phase space and that in terms of quantum Sobolev spaces, they exhibit the same maximal regularity as characteristic functions. This can be interpreted as a semiclassical asymptotic on the size of commutators in Schatten norms. Our study answers a question raised in [J.~Chong, L.~Lafleche, C.~Saffirio, arXiv:2103.10946 [math.AP]] about the possibility of having projection operators as initial data. It also gives a strong convergence result in Sobolev spaces for the Weyl law in phase space.
\end{abstract}

\begingroup
\def\uppercasenonmath#1{} 
\let\MakeUppercase\relax 
\maketitle
\endgroup

\bigskip


\renewcommand{\contentsname}{\centerline{Contents}}
\setcounter{tocdepth}{2}	
\tableofcontents

\section{Introduction}

	Projection operators arise naturally in quantum mechanics, the simplest example being the projection operator $\op = \ket{\psi}\bra{\psi}$ associated to a wave function $\psi : \Rd\to \CC$ verifying $\intd \n{\psi}^2 = 1$, which is defined by
	\begin{equation*}
		\op \varphi(x) = \psi(x) \intd \conj{\psi(y)}\,\varphi(y)\d y
	\end{equation*}
	for any $\varphi\in L^2(\Rd)$. More generally, a projection operator can be defined as an operator $\op$ verifying $\op^2 = \op$. In this paper, we will consider compact operators acting on $L^2(\Rd)$, and we denote by $\L^\infty$ the set of such operators. It is not difficult to see that such an operator is automatically trace class, and so if $\op$ is self-adjoint, it can be thought of as a density operator representing the state of a quantum system. 
	
	In the classical limit, that is in units where the Planck constant $h = 2\pi\,\hbar$ becomes negligible, the classical analogue of a density operator is a phase space density $f(x,v)$ representing the probability to find a particle at position $x\in\Rd$ with velocity $v\in\Rd$. With this point of view, projection operators can be thought of as the quantum analogue of functions verifying $f^2 = f$, that is characteristic functions. It is well-known (see e.g.~\cite{sickel_regularity_2021}) that characteristic functions cannot be infinitely regular, and the maximal regularity allowed for test functions can be formulated by saying that such functions are in the Besov space $B^s_{p,\infty}$ at most for $s \leq 1/p$. It is the goal of this paper to prove that an analogue of this property holds true for projections operators.
	
	\subsection{Phase space quantum mechanics} To make the analogy between density operators and phase space functions more precise, it is typical to introduce the Weyl quantization, which by analogy with the Fourier inversion formula associates to the function $f$ the operator
	\begin{equation}\label{eq:Weyl_def_0}
		\op_f := \intdd \widehat{f}(y,\xi) \,e^{2i\pi\(y\cdot x + \xi\cdot \opp\)}\d y \d\xi
	\end{equation}
	with $x$ identified with the operator of multiplication by $x$ and $\opp = -i\hbar\nabla$ the quantum analogue of the momentum, where $\nabla$ denotes the gradient with respect to the $x$ variable and $\hbar = h/(2\pi)$ is the Planck constant. The integral kernel of this operator is then given by
	\begin{equation*} 
		\op_f(x,y) = \intd e^{-2i\pi\(y-x\)\cdot\xi} \, f(\tfrac{x+y}{2},h\xi)\d\xi.
	\end{equation*}
	One can also look at the inverse operation called the Wigner transform, that associates to an operator $\op$ a function on the phase space
	\begin{equation*} 
		f_{\op}(x,\xi) = \intd e^{-i\,y\cdot\xi/\hbar} \,\op(x+\tfrac{y}{2},x-\tfrac{y}{2})\d y.
	\end{equation*}
	Noticing for example that at least formally $h^d\tr(\op_f) = \intdd f(x,\xi)\d x\d\xi$ where $\tr$ denotes the trace, and that $h^d\tr(|\op_f|^2) = \intdd \n{f(x,\xi)}^2\d x\d\xi$, where $\n{A} = \sqrt{A^*A}$ denotes the absolute value of an operator, it is natural to consider the following scaled Schatten norms
	\begin{equation}\label{eq:def_norm}
		\Nrm{\op}{\L^p} = h^{\frac{d}{p}} \Nrm{\op}{p} = h^{\frac{d}{p}} \Tr{\n{\op}^p}^\frac{1}{p},
	\end{equation}
	that are the quantum version of the phase space Lebesgue norms. On another side, the correspondence principle leads to define the quantum analogue of the gradients in the phase space by the following formulas
	\begin{equation}\label{eq:quantum_gradients}
		\Dhx \op := \com{\nabla,\op} \quad \text{ and } \quad \Dhv \op := \com{\frac{x}{i\hbar},\op}.
	\end{equation}
	These formulas can also be understood as the Weyl quantization of the classical phase space gradients since
	\begin{equation}\label{eq:weyl_quantization_gradients}
		\op_{\Dx f} = \Dhx{\op_f} \quad \text{ and } \quad \op_{\Dv f} = \Dhv{\op_f}.
	\end{equation}
	The uniform-in-$\hbar$ boundedness of these quantities in the scaled Schatten norms can thus be seen as the quantum analogue of the boundedness of a phase space function to a Sobolev space. We refer to \cite{lafleche_quantum_2022} for a more detailed presentation of these ideas and their applications.
	
\subsection{Motivation: Slater determinants and semiclassical mean-field limit}

	We are interested by one-particle density matrices that are projections operators. These states appear for instance when considering the one-particle reduced density of a Slater determinant. Recall that a Slater determinant $\omega_N$ can be defined as the $N$-body wave function
	\begin{equation}\label{eq:Slater}
		\Psi_N(x_1,\dots,x_N) := \frac{1}{\sqrt{N!}} \det(\psi_j(x_k))_{(j,k)\in\set{1,\dots,N}^2}
	\end{equation}
	or the associated density operator
	\begin{equation*}
		\omega_N = \ket{\Psi_N}\bra{\Psi_N}
	\end{equation*}
	where $(\psi_1, \dots,\psi_N)$ is an orthonormal family of $L^2(\Rd)$. Its one-particle reduced density, or first marginal, is then of the form
	\begin{equation*}
		\omega = N \tr_{2,\dots,N}(\omega_N) = \sum_{j=1}^{N} \ket{\psi_j}\bra{\psi_j}
	\end{equation*}
	if we choose the normalization $\Tr{\omega} = N$. In particular, it remains a projection, i.e. it verifies $\omega^2 = \omega$. Reciprocally, to any self-adjoint one-particle density operator $\omega$ verifying 
	\begin{equation}\label{eq:def_pure_state}
		\Tr{\omega} = N \quad \text{ and } \quad \omega^2 = \omega
	\end{equation}
	one can associate a Slater determinant~\eqref{eq:Slater} by the spectral theorem. To be compatible with the Weyl quantization and the Wigner transform and see what is happening when units are chosen so that $\hbar$ becomes negligible, or equivalently in the classical limit $\hbar\to 0$, we define $\op = (Nh^d)^{-1}\omega$ so that $h^d\Tr{\op} = \iintd f_{\op} = 1$, and we assume that
	\begin{equation*}
		\Nrm{f_{\op}}{L^2(\Rdd)} = \cC_2 < \infty
	\end{equation*}
	converges to a constant when $\hbar\to 0$. By the properties of the Wigner transform it holds
	\begin{equation*}
		\cC_2^2 = h^d\Tr{\op^2} = \frac{1}{N^2 h^d} \Tr{\omega^2} = \frac{1}{N^2 h^d}\Tr{\omega} = \frac{1}{N h^d}.
	\end{equation*}
	Therefore, $N$ and $h$ are linked through the relation $N h^d = \cC_2^{-2}$. This is in contrast with the case of fermionic mixed states where more generally, it is possible to have $N\,h^d \leq \cC_2^{-2}$ (see \cite{chong_many-body_2021}). Observe additionally that it follows from Equation~\eqref{eq:def_pure_state} that the operator norm of $\omega$ is given by $\Nrm{\omega}{\infty} = 1$. Hence, by definition, the operator norm of $\op$ is given by $\Nrm{\op}{\L^\infty} = \Nrm{\op}{\infty} = \frac{1}{N\,h^d} = \cC_2^2$, and is also independent of $\hbar$. To summarize $\cC_2^2 = (N h^d)^{-1}$ is independent of $\hbar$ and such that $\op = \cC_2^2\,\omega$. To simplify, we will just consider in the rest of the paper that $\cC_2 = 1$, so that the theorems would correspond to the case $h = N^{-1/d}$ and $\op^2 = \op = \omega$.
	
	One of the motivation of this investigation arises from the problem of the mean-field and semiclassical limit from the $N$-body Schr\"odinger equation to the Hartree--Fock and Vlasov equations when the interaction potential is singular, as studied in~\cite{porta_mean_2017, chong_many-body_2021}. In the second of these works, assumptions of \textit{semiclassical regularity} are made on the initial data in the sense of the quantum Sobolev spaces defined in Section~\ref{sec:Sobolev}, and we prove here that these assumptions are not compatible with projection operators, and so Slater determinants, as was conjectured in~\cite[Remark~4.3]{chong_many-body_2021}.
	
	Precise estimates on the commutator of operators with the operators $x$ and $\opp$ are also useful to understand the size of the self-distance of some pseudometrics in quantum optimal transport, which is related to the Wigner--Yanase Skew information (See~\cite{de_palma_quantum_2021-1, lafleche_quantum_2023}).

\section{Main results}

	Our results consider three main cases. The first result, Theorem~\ref{thm:noregu}, is a result about the lack of high regularity for projection operators and is valid for all projection operators. Then, Theorem~\ref{thm:regu} indicates that there are states verifying the maximal regularity allowed by the first theorem. Finally, Theorem~\ref{thm:regu_noregu} claims that there are however states for which the regularity is strictly lower than the maximal regularity. It also shows as an application that the regularity obtained for projection operators can lead to improvements in the statement of the Weyl law.
	
\subsection{Sobolev spaces}\label{sec:Sobolev}
	
	Denote by $z = (x,\xi) \in\Rdd$ the phase space variable. Then, in the classical setting, the homogeneous Sobolev space of order $1$ of functions on the phase space can be defined as the set of functions $f : \Rdd \to \R$ vanishing at infinity and for which the following norm is finite
	\begin{equation*}
		\Nrm{f}{\dot{W}^{1,p}(\Rdd)} = \Nrm{\nabla_z f}{L^p(\Rdd)}.
	\end{equation*}
	As usual, one can also define the corresponding non-homogeneous space by defining the norm $\Nrm{f}{W^{1,p}(\Rdd)} = \Nrm{f}{L^p(\Rdd)} + \Nrm{\nabla_z f}{L^p(\Rdd)}$. Analogously, the quantum Sobolev norms of order $1$ are defined by the formula
	\begin{equation*}
		\Nrm{\op}{\dot{\cW}^{1,p}} := \Nrm{\Dh\op}{\L^p}.
	\end{equation*}
	where $\Dh\op$ is the vector valued operator $(\Dhx\op,\Dhv\op)$ whose absolute value verifies $\n{\Dh\op}^2 = \n{\Dhx\op}^2 + \n{\Dhv{\op}}^2$. When $s\in(0,1)$, the quantum analogue of the fractional Gagliardo--Sobolev norms are defined by $\Nrm{\op}{\cW^{s,p}} = \Nrm{\op}{\L^p} + \Nrm{\op}{\dot{\cW}^{s,p}}$ in~\cite{lafleche_quantum_2022} where
	\begin{equation}\label{eq:Sobolev_frac_def}
		\Nrm{\op}{\dot{\cW}^{s,p}}^p :=  \gamma_{s,p}\, h^d\intdd \frac{\Tr{\n{\sfT_z\op - \op}^p}}{\n{z}^{2d+sp}}\d z,
	\end{equation}
	for some constant $\gamma_{s,p}>0$. Here, $\sfT_z$ is the quantum phase space translation operator defined by
	\begin{equation}\label{eq:translation_def}
		\sfT_{z_0}\op = e^{i\(\xi_0\cdot x -x_0\cdot\opp\)/\hbar}\, \op\, e^{i\(x_0\cdot\opp - \xi_0\cdot x\)/\hbar}.
	\end{equation}
	More generally when $s\in[0,2)$, the quantum Besov norms~\cite{lafleche_quantum_2022} are defined by $\Nrm{\op}{\cB^s_{p,q}} = \Nrm{\op}{\L^p} + \Nrm{\op}{\dot{\cB}^s_{p,q}}$ where
	\begin{equation}\label{eq:Besov_def}
		\Nrm{\op}{\dot{\cB}^s_{p,q}} := \BNrm{\frac{\Nrm{\sfT_{2z}\op - 2\,\sfT_z\op + \op}{\L^p}}{\n{z}^{s+2d/q}}}{L^q(\Rdd)}.
	\end{equation}
	We will write\footnote{More rigorously, we could define a space $\ell^\infty\cB^s_{p,q}$ of $\hbar$ dependent operators of the form $\op = (\op_\hbar)_{\hbar\in(0,1)}$ and define the norm as $\Nrm{\op}{\ell^\infty\cB^s_{p,q}} = \sup_{\hbar\in(0,1)} \Nrm{\op_\hbar}{\cB^s_{p,q}}$, so that the notation~\eqref{eq:semiclassical_space} would be replaced by $\op \in \ell^\infty\cB^s_{p,q}$. The sequence of inclusions~\eqref{eq:Besov_inclusions} should also be understood from this point of view.}
	\begin{equation}\label{eq:semiclassical_space}
		\op \in \cB^s_{p,q}
	\end{equation}
	whenever there exists a constant $C$ independent of $\hbar$ such that $\Nrm{\op}{\cB^s_{p,q}} < C$. These Besov norms yield a finer scale compared to Sobolev norms, as they verify for any $s\in(0,1)$, any real numbers $1< r < p < q$ and any $\eps > 0$ sufficiently small (see~\cite{lafleche_quantum_2022})
	\begin{equation}\label{eq:Besov_inclusions}
		\cB^{s + \eps}_{p,q} \subset \cB^{s}_{p,1} \subset \cB^{s}_{p,r} \subset \cW^{s,p} \subset \cB^{s}_{p,q} \subset \cB^{s}_{p,\infty} \subset \cB^{s - \eps}_{p,q}.
	\end{equation}
	
\subsection{Maximal regularity of projections}

	Our main result can be summarized by telling that if $\op$ is a compact projection operator in $\L^1$, then $\op \notin \cB^s_{p,q}$ whenever $s > 1/p$, or $s= 1/p$ and $q < \infty$. This corresponds to the red part in Figure~\ref{fig:1}.

	\begin{thm}\label{thm:noregu}
		Let $s\in[0,1]$, $(p,q)\in[1,\infty]^2$ and $\hbar$ be some sequence converging to~$0$. Then, if $(\op_\hbar)_{\hbar\in(0,1)}$ is a sequence of operators such that $\op_\hbar^2=\op_\hbar$ and $h^d\Tr{\op}=1$, the following holds.
		\begin{itemize}
		\item If the operators are self-adjoint, then
		\begin{align}\label{eq:Besov_noregu}
			\Nrm{\op_\hbar}{\dot{\cB}^s_{p,q}}\ &\to \infty &&\text{ if } s > 1/p \text{, or } s= 1/p \text{ and } q < \infty
			\\\label{eq:Sobolev_noregu}
			\Nrm{\op_\hbar}{\dot{\cW}^{s,p}} &\to \infty &&\text{ if } s \geq 1/p \text{ and } p > 1.
		\end{align}
		\item If the operators are not self-adjoint, then the result still holds if $p \geq \frac{2\,d}{d+s}$. If $p < \frac{2\,d}{d+s}$, then the result holds if $\op\in\L^{2+\eps}$ uniformly in $\hbar$ for some $\eps>0$.
		\end{itemize}
	\end{thm}
	
	Using the above theorem, and more specifically Equation~\eqref{eq:Sobolev_noregu}, in the case of integer order of regularity gives the behavior of some commutators in Schatten norms.
	
	\begin{cor}
		For any $p>1$ and any sequence of self-adjoint projection operators $\op$ bounded in $\L^1$ uniformly in $\hbar$
		\begin{equation*}
			\frac{1}{\hbar} \Nrm{\com{x,\op}}{\L^p} + \Nrm{\com{\nabla,\op}}{\L^p} \underset{\hbar\to 0}{\to} \infty.
		\end{equation*}
		In the particular case when $\op\in\dot{\cB}^s_{p,q}$ for some $s>0$ and $p\geq 2$, then it follows from the proof of Proposition~\ref{prop:cv_pure_states} that both of the quantities in the above equation actually tend to $\infty$ separately.
	\end{cor}
	
	Equation~\eqref{eq:Besov_noregu} actually also gives the semiclassical behavior of some commutators, even in the case of a non-integer order of regularity. This follows by noticing that $\sfT_z-1 = \sfT_{(x,0)}(\sfT_{(0,\xi)}-1) + \sfT_{(x,0)}-1$, and
	\begin{equation}\label{eq:Besov_vs_commutators}
		\n{\sfT_{(0,\xi)}\op - \op} = \n{e^{i\,x\cdot\xi/\hbar} \,\op \,e^{-i\,x\cdot \xi/\hbar} - \op} = \n{\com{e^{i\,x\cdot\xi/\hbar},\op}}
	\end{equation}
	and an analogous formula holds for $\sfT_{(x,0)}\op - \op$. This gives the following asymptotic result.
	
	\begin{cor}\label{cor:com_exp}
		With the same hypotheses as in Theorem~\ref{thm:noregu}, for any $s > 1/p$,
		\begin{equation*}
			\sup_{(y,\xi)\in\Rdd} \frac{1}{\hbar^s\n{\xi}^s} \Nrm{\com{e^{2i\pi\,x\cdot\xi},\op}}{\L^p} + \frac{1}{\hbar^s\n{y}^s} \Nrm{\com{e^{2i\pi\,y\cdot\opp},\op}}{\L^p} \to \infty.
		\end{equation*}
		As in the previous corollary, in the particular case when $\op\in\dot{\cB}^s_{p,q}$ for some $s>0$ and $p\geq 2$, both terms tend to $\infty$ separately.
	\end{cor}
	
\subsection{The case of Schrödinger operators}

	The Theorem~\ref{thm:noregu} is optimal as there are examples of operators for which the maximal allowed  regularity is reached. A particular class of states are the one considered in~\cite{fournais_optimal_2020}, which are of the form
	\begin{equation}\label{eq:spectral_proj}
		\op = \indic_{(-\infty,0]}(-\hbar^2\Delta + V(x)) =: \indic_{\n{\opp}^2 \leq U(x)}
	\end{equation}
	where $U = -V$ is such that there exists $\eps > 0$ and open sets $\Omega_\eps$ and $\Omega$ verifying $\overline{\Omega_\eps} \subset \Omega\subset\Rd$ such that
	\begin{system}\label{eq:conditions_u}
		&U\in C^\infty(\Omega) \cap L^1_\loc(\Omega^c)
		\\
		&U \leq -\eps \text{ on }\Omega_\eps^c.
	\end{system}
	Then it follows from~\cite[Theorem~1.2]{fournais_optimal_2020} that there exists a constant $C > 0$ independent of $\hbar\in(0,1)$ such that
	\begin{equation}\label{eq:soren_bound}
		\Nrm{\Dhx{\op}}{\L^1} \leq C  \qquad \text{ and } \qquad \Nrm{\Dhv{\op}}{\L^1} \leq C.
	\end{equation}
	More precisely, this implies the following.
	\begin{thm}\label{thm:regu}
		Let $\op = \indic_{\n{\opp}^2 \leq U(x)}$ or $\op = \indic_{\n{x}^2 \leq U(\opp)}$ with $U$ verifying Assumptions~\eqref{eq:conditions_u} and $U_+\in L^{d/2}(\Rd)$. Then $\op^2 = \op$ and for any $p\in[1,\infty]$,
		\begin{equation}\label{eq:Besov_regu}
			\op\in \cB^{1/p}_{p,\infty} \cap \cW^{1,1}.
		\end{equation}
		uniformly in $\hbar$. In particular, the Wigner transform of $\op$ verifies $f_{\op}\in B^{1/2}_{2,\infty}(\Rdd)$ uniformly in $\hbar$, hence, for any $s \in[0,1/2)$ 
		\begin{equation}\label{eq:Sobolev_regu}
			f_{\op} \in H^s(\Rdd)
		\end{equation}
		uniformly in $\hbar$.
	\end{thm}

	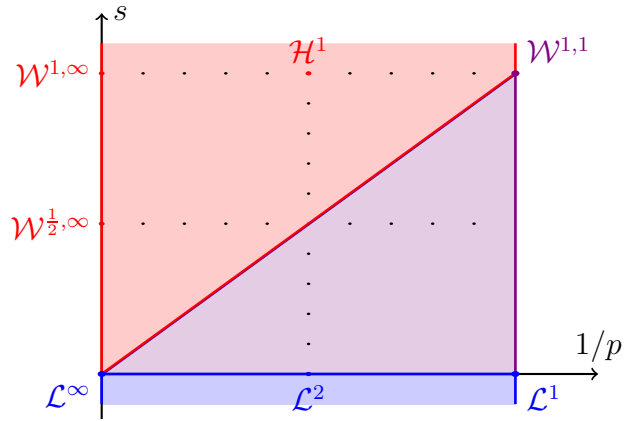
\begin{figure}[ht]
	\begin{tikzpicture}[yscale=0.4,xscale=0.55]
	
	\fill[fill=red!100, opacity=0.2]
	(0,0)--(10,10)--(10,11)--(0,11)--(0,10)
	-- cycle;
	\fill[fill=violet!100, opacity=0.2]
	(0,0)--(10,10)--(10,0)
	-- cycle;
	\fill[fill=blue!100, opacity=0.2]
	(0,0)--(10,0)--(10,-1)--(0,-1)
	-- cycle;

	
	\draw[->,line width = 0.8pt](0,0)--(12,0) node[above]{$1/p$};
	
	\draw[->,line width = 0.8pt](0,-1.5)--(0,12) node[right]{$s$};
	
	\foreach \x in {5}
		\foreach \n in {1, ..., 9}
			\fill (\x,\n) circle [radius=1.2pt];
	\foreach \n in {5,10}
		\foreach \x in {1, ..., 4, 6, 7, 8, 9}
			\fill (\x,\n) circle [radius=1.2pt];

	
	\draw [line width = 1, red] (0,11)--(0,0)--(10,10)--(10,11) node[midway, color=red, align=center]{\footnotesize };
	\draw [line width = 0.4, violet] (0,-0.05)--(10,9.95) node[midway, color=red, align=center]{\footnotesize };
	\draw [line width = 1, blue] (0,-1)--(0,0)--(10,0)--(10,-1) node[midway, right, color=blue, align=center]{\footnotesize }; 
	\draw [line width = 1, violet] (10,0)--(10,10) node[midway, above left, color=blue, align=center]{\scriptsize }; 

	
	\fill[red] (0,10) circle [radius=2pt] node[left, color=red]{$\cW^{1,\infty}$};
	\fill[red] (5,10) circle [radius=2pt] node[above, color=red]{$\cH^1$};
	\fill[violet] (10,10) circle [radius=2.8pt] node[above right, color=violet]{$\cW^{1,1}$};
	\fill[red] (0,5) circle [radius=2pt] node[left, color=red]{$\cW^{\frac{1}{2},\infty}$};
	\fill[blue] (0,0) circle [radius=2.8pt] node[below left, color=blue]{$\L^\infty$};
	\fill[blue] (5,0) circle [radius=2pt] node[below, color=blue]{$\L^2$};
	\fill[blue] (10,0) circle [radius=2.8pt] node[below right, color=blue]{$\L^1$};
	\end{tikzpicture}
	\caption{Regularity of self-adjoint projection operators in Sobolev spaces $\cW^{s,p}$. The top (red) zone corresponds to forbidden spaces, the middle (violet) zone corresponds to allowed spaces, the (blue) region below corresponds to the spaces where projection operators always are. The diagonal line corresponds to the forbidden spaces $\cW^{1/p,p}$ with $p\in(1,\infty)$, slightly smaller than the allowed spaces $\cB^{1/p}_{p,\infty}$.}\label{fig:1}
	\end{figure}
	
	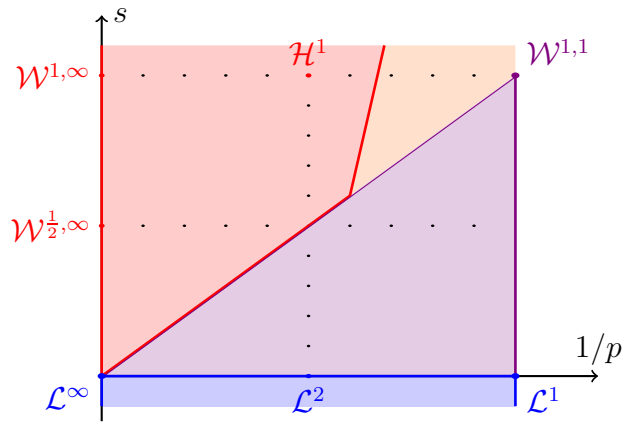
\begin{figure}[ht]
	\begin{tikzpicture}[yscale=0.4,xscale=0.55]
	
	\fill[fill=red!100, opacity=0.2]
	(0,0)--(6,6)--(41/6,11)--(0,11)--(0,10)
	-- cycle;
	\fill[fill=orangeFonce!100, opacity=0.2]
	 (6,6)--(10,10)--(10,11)--(41/6,11)
	-- cycle;
	\fill[fill=violet!100, opacity=0.2]
	(0,0)--(10,10)--(10,0)
	-- cycle;
	\fill[fill=blue!100, opacity=0.2]
	(0,0)--(10,0)--(10,-1)--(0,-1)
	-- cycle;

	
	\draw[->,line width = 0.8pt](0,0)--(12,0) node[above]{$1/p$};
	
	\draw[->,line width = 0.8pt](0,-1.5)--(0,12) node[right]{$s$};
	
	\foreach \x in {5}
		\foreach \n in {1, ..., 9}
			\fill (\x,\n) circle [radius=1.2pt];
	\foreach \n in {5,10}
		\foreach \x in {1, ..., 4, 6, 7, 8, 9}
			\fill (\x,\n) circle [radius=1.2pt];

	
	\draw [line width = 1, red] (0,11)--(0,0)--(6,6)--(41/6,11) node[midway, color=red, align=center]{\footnotesize };
	\draw [line width = 0.4, violet] (0,-0.05)--(10,9.95) node[midway, color=red, align=center]{\footnotesize };
	\draw [line width = 1, blue] (0,-1)--(0,0)--(10,0)--(10,-1) node[midway, right, color=blue, align=center]{\footnotesize }; 
	\draw [line width = 1, violet] (10,0)--(10,10) node[midway, above left, color=blue, align=center]{\scriptsize };

	
	\fill[red] (0,10) circle [radius=2pt] node[left, color=red]{$\cW^{1,\infty}$};
	\fill[red] (5,10) circle [radius=2pt] node[above, color=red]{$\cH^1$};
	\fill[violet] (10,10) circle [radius=2.8pt] node[above right, color=violet]{$\cW^{1,1}$};
	\fill[red] (0,5) circle [radius=2pt] node[left, color=red]{$\cW^{\frac{1}{2},\infty}$};
	\fill[blue] (0,0) circle [radius=2.8pt] node[below left, color=blue]{$\L^\infty$};
	\fill[blue] (5,0) circle [radius=2pt] node[below, color=blue]{$\L^2$};
	\fill[blue] (10,0) circle [radius=2.8pt] node[below right, color=blue]{$\L^1$};
	\end{tikzpicture}
	\caption{Regularity of non self-adjoint projection operators in Sobolev spaces $\cW^{s,p}$ in dimension $d=3$. In this case, Theorem~\ref{thm:noregu} needs more assumptions than a trace condition to be valid in the top-right zone. We conjecture however that the result holds there as well.}
	\end{figure}
	
	\begin{remark}
		It might seem surprising that one can find projection operators in the space $\cW^{1,1}$ while a classical characteristic function $\chi(z) = \indic_{A}(z)$ of a nonempty set $A\subset \Rdd$ is never in $W^{1,1}(\Rdd)$. But such a function $\chi$ can be in $BV(\Rdd)$, the set of distributions with bounded variation, that is the distributions whose gradient is a measure. In semiclassical analysis, the set $\L^1$ should indeed rather be interpreted as the quantum analogue of the space of measures. As an example, the Weyl quantization of the function $g_h(z) = \(2/h\)^d e^{-\n{z}^2/\hbar}$ verifies for any $h>0$, $\bNrm{\op_{g_h}}{\L^1} = 1$, but its Wigner transform converges to $\delta_0$. Similarly, the set $\cW^{1,1}$ should be rather interpreted as the quantum analogue of $BV(\Rdd)$.
	\end{remark}
	
	\begin{remark}
		It follows from the definition of the quantum gradients~\eqref{eq:quantum_gradients} that for any $f\in H^1(\Rdd)$ and any $\varphi\in C^1(\Rd)$,
		\begin{equation*}
			\Nrm{\com{\varphi(x),\op_f}}{\L^2}^2 = h^d\intdd \n{\(\varphi(x)-\varphi(y)\)\op_f(x,y)}^2 \d x\d y \leq \hbar^2 \Nrm{\varphi}{C^1(\Rd)}^2 \Nrm{\Dhv{\op_f}}{\L^2}^2.
		\end{equation*}
		Since $\sNrm{\Dhv{\op_f}}{\L^2} \leq \sNrm{\op_f}{\cW^{1,2}} = \Nrm{f}{H^1(\Rdd)}$, and on the other hand
		\begin{equation*}
			\Nrm{\com{\varphi(x),\op_f}}{\L^2} \leq 2 \Nrm{\varphi}{L^\infty(\Rd)} \sNrm{\op_f}{\L^2} = 2 \Nrm{\varphi}{L^\infty(\Rd)} \Nrm{f}{L^2(\Rdd)}
		\end{equation*}
		it follows by bilinear complex interpolation (see~\cite[Lem~28.3]{tartar_introduction_2007}) that
		\begin{equation*}
			\Nrm{\com{\varphi(x),\op_f}}{\L^2} \leq \sqrt{C_d\hbar} \Nrm{\varphi}{B^{1/2}_{\infty,1}(\Rd)} \Nrm{f}{B^{1/2}_{2,\infty}(\Rdd)}
		\end{equation*}
		for some constant $C_d$ depending only on $d$. Therefore, with the notations of \cite{deleporte_universality_2021}, applying the above inequality to $f = f_{\op}$,
		\begin{equation}\label{eq:conjecture_0}
			2\,h^d\, \mathrm{Var}(X(\varphi)) = h^d\Tr{\n{\com{\varphi(x),\op}}^2} \leq C_d\,\hbar \Nrm{\varphi}{B^{1/2}_{\infty,1}(\Rd)}^2 \Nrm{\op}{\cB^{1/2}_{2,\infty}}^2
		\end{equation}
		where we recall that $\hbar = h/(2\pi)$. Equation~\eqref{eq:conjecture_0} gives a quantitative version of Formula~(5.30) in \cite{deleporte_universality_2021} about the variance of the linear statistics of the determinantal process associated to $\op$. In this paper, Deleporte and Lambert conjecture that a lower bound of the same order should hold for any $\varphi\in C^\infty_c$ nonzero on the set $\set{U>0}$, that is, if $\op$ is of the form \eqref{eq:spectral_proj} there would exists a constant $c_\varphi>0$ independent of $\hbar$ such that
		\begin{equation*}
			\Nrm{\com{\varphi(x),\op}}{\L^2} \geq c_\varphi\,\sqrt{\hbar}.
		\end{equation*}
		Notice that it already follows from Corollary~\ref{cor:com_exp} and from Theorem~\ref{thm:regu} that for any $\eps >0$ and any $\op$ of the form~\eqref{eq:spectral_proj} with $U$ verifying \eqref{eq:conditions_u}
		\begin{equation*}
			\sup_{\xi\in\Rd} \Nrm{\com{\tfrac{e^{2i\pi\,x\cdot\xi}}{\n{\xi}^{1/2-\eps}},\op}}{\L^2} \geq G(\hbar)\, \hbar^{1/2-\eps}
		\end{equation*}
		for some function $G$ such that $G(\hbar) \underset{\hbar\to 0}{\rightarrow} \infty$.
	\end{remark}
	
	Thanks to the well-known Weyl asymptotic formulas, we can have more precise results for spectral projections of the form~\eqref{eq:spectral_proj} (see e.g. \cite{fournais_semi-classical_2018}). It implies the following theorem which indicates that the regularity of Theorem~\ref{thm:regu} is not generic, and that there are projection operators that are less regular. More precisely, for any $s>0$ there are projections operators that are not bounded in $\cW^{s,p}$ uniformly in $\hbar$.
	\begin{thm}\label{thm:regu_noregu}
		Let $\op_\hbar = \indic_{\n{\opp}^2\leq U(x)}$ with $U_+\in L^\infty(\Rd)\cap L^{d/2}(\Rd)$. Then the Husimi transform of $\op_\hbar$ converges weakly in $L^p(\Rdd)$ to $\indic_{\n{\xi}^2\leq U(x)}$ for all $p\in(0,1)$ when $\hbar\to 0$. Moreover, when $U$ satisfies Assumption~\eqref{eq:conditions_u} with $\Omega$ bounded, then for any $p\in[1,\infty)$ and $s\in[0,1/p)$,
		\begin{align}\label{eq:CV_Sobolev_Husimi}
			\tilde{f}_{\op_\hbar} &\underset{\hbar\to 0}{\to} \indic_{\n{\xi}^2 \leq U(x)} \quad \text{ in } W^{s,p}(\Rdd)
			\\\label{eq:CV_Sobolev_Wigner}
			f_{\op_\hbar} &\underset{\hbar\to 0}{\to} \indic_{\n{\xi}^2 \leq U(x)} \quad \text{ in } W^{s,p}(\Rdd) \text{ with } s\in \(\tfrac{1}{p'},\tfrac{1}{2} + 2d\(\tfrac{1}{p}-\tfrac{1}{2}\)\)
		\end{align}
		where $p' = \frac{p}{p-1}$. There is also convergence in $L^p(\Rdd)$ for any $p\in[1,\infty)$.
		
		However, there exists functions $U$ with $\sqrt{U}\in L^\infty(\Rd)\cap L^d(\Rd)\cap C^\alpha(\Rd)$  with $\alpha\in(0,1)$ such that for all $s > \alpha/p$, $\indic_{\n{\xi}^2 \leq U(x)} \notin B^s_{p,q}(\Rdd)$ and so such that
		\begin{equation}\label{eq:existence_of_irregularity}
			\Nrm{\op_\hbar}{\cB^s_{p,q}} \underset{\hbar\to 0}{\to} \infty.
		\end{equation}
		
	\end{thm}

\section{The explicit example of the Harmonic oscillator.}

	An example where one can do explicit computations is the case of non-interacting fermions in a harmonic trap, see also for example~\cite{benedikter_effective_2022}. Consider the Hamiltonian given by the harmonic oscillator
	\begin{equation*}
		\sfH_N = \dG(H) = \sum_{n=1}^N H_n \quad \text{ where } \quad H = \n{\opp}^2 + \n{x}^2
	\end{equation*}
	and where each $H_n$ is the one particle Hamiltonian $H$ acting on the $n$-th particle. It is well-known that the eigenvalues of $H$ are given of the form $\(2\,n+1\)\hbar$ with $n\in\N$, with eigenvectors $\psi_n$ given by Hermite functions. The ground state of $\sfH_N$ is given by a Slater determinant formed by the $N$ first eigenvectors of the one-body Hamiltonian. Assuming for simplicity that the number of particles is such that $N = \binom{d+n}{d}$ with $n\in\N$, then its one-particle density operator can be written
	\begin{equation*}
		\op = \sum_{\n{\alpha}_1 \leq n} \ket{\psi_\alpha}\bra{\psi_\alpha}
	\end{equation*}
	where $\alpha = (\alpha_1,\dots\alpha_d)\in\N_0^d$ is a multi-index, with $\n{\alpha}_1 = \alpha_1 + \dots + \alpha_d$ and $\psi_\alpha = \psi_{\alpha_1}\otimes\dots\otimes\psi_{\alpha_d}$. It verifies $\Tr{\op} = N = h^{-d}$ and $0\leq \op\leq 1$. Following similar computations as the one done for example in \cite{benedikter_effective_2022}, it holds
	\begin{equation}\label{eq:gradient_ground_state}
		\n{\Dh_{\xi_1}{\op}}^p = \frac{1}{(2\hbar)^{p/2}} \sum_{\n{\alpha}_1 = n}  \(\alpha_1+1\)^{p/2}\big(\ket{\psi_{\alpha}}\bra{\psi_{\alpha}} + \ket{\psi_{\alpha+e_1}}\bra{\psi_{\alpha+e_1}}\,\big).
	\end{equation}
	From this we deduce that for any $p\in[1,\infty]$, there exists $C_p$ independent of $\hbar$ such that
	\begin{equation*}
		\Nrm{\Dh_{\xi_1}\op}{\L^p} = \frac{C_p}{h^{1/p'}}
	\end{equation*}
	where $p' = \frac{p}{p-1}$ is the H\"older conjugate of $p$. The constant $C_p$ can be explicitly computed and verifies $C_p \leq (2d)^{1/p}\sqrt{\pi}(d!)^{\frac{1}{d}\(\frac{1}{2}-\frac{1}{p}\)}$ when $p\leq 2$ and $C_p \leq 2^{1/p}\sqrt{\pi}(d!)^{\frac{1}{d}\(\frac{1}{2}-\frac{1}{p}\)}$ when $p\geq 2$. Moreover,
	\begin{align*}
		\Nrm{\Dh_{\xi_1}\op}{\L^1} &\leq \frac{2d\sqrt{\pi}}{\sqrt{d!}}, & \Nrm{\Dh_{\xi_1}\op}{\L^2} &= \frac{1}{\sqrt{\hbar}}, & \Nrm{\Dh_{\xi_1}\op}{\L^\infty} &\leq \frac{\sqrt{(d!)^{1/d}\,\pi}}{h}.
	\end{align*}
	The same estimates hold by replacing $\Dhv{}$ by $\Dhx{}$. In particular, $\op\in\cW^{1,1}\cap \L^\infty$ and so in the same way as in the proof of Theorem~\ref{thm:regu}, $\op\in \cB^{1/p}_{p,\infty}$ for any $p\in[1,\infty]$.
	
	\begin{proof}
		 Following similar computations as in~\cite{benedikter_effective_2022}, one can compute $\com{x_1,\op}$ using the fact that $x_1 = \frac{a+a^*}{2}$ where $a^* = x-i\opp$ and $a = x+i\opp$ are the creation and annihilation operators. This yields
		\begin{equation*}
			\com{x_1,\op} = \frac{\sqrt{\hbar}}{\sqrt{2}} \sum_{\n{\tilde{\alpha}}_1 \leq n} \sqrt{a+1}\,\big(\ket{\psi_{a+1}}\bra{\psi_a} - \ket{\psi_a}\bra{\psi_{a+1}}\big)\otimes \ket{\psi_{\tilde{\alpha}}}\bra{\psi_{\tilde{\alpha}}}
		\end{equation*}
		where $\tilde{\alpha} = (0,\alpha_2,\dots,\alpha_d)$ and $a = n-\n{\tilde{\alpha}}_1$. By taking $\alpha = (\alpha_1, \alpha_2,\dots,\alpha_d)$ with $\alpha_1 = a$, this can be equivalently written
		\begin{equation*}
			\Dh_{\xi_1}{\op} = \frac{i}{\sqrt{2\hbar}} \sum_{\n{\alpha}_1 = n}  \sqrt{\alpha_1+1}\,\big(\ket{\psi_{\alpha_1}}\bra{\psi_{\alpha_1+1}} - \ket{\psi_{\alpha_1+1}}\bra{\psi_{\alpha_1}}\big)\otimes \ket{\psi_{\tilde{\alpha}}}\bra{\psi_{\tilde{\alpha}}}.
		\end{equation*}
		Notice that $\psi_{\alpha_1}\otimes\psi_{\tilde{\alpha}} = \psi_{\alpha}$ and $\psi_{\alpha_1+1}\otimes\psi_{\tilde{\alpha}} = \psi_{\alpha+e_1}$ with $e_1=(1,0,\dots,0)$. Hence, taking the square of the above operator and using the fact that the family of $\psi_n$ is orthonormal leads to
		\begin{equation*}
			\n{\Dh_{\xi_1}{\op}}^2 = \frac{1}{2\hbar} \sum_{\n{\alpha}_1 = n}  \(\alpha_1+1\)\big(\ket{\psi_{\alpha}}\bra{\psi_{\alpha}} + \ket{\psi_{\alpha+e_1}}\bra{\psi_{\alpha+e_1}}\,\big).
		\end{equation*}
		As this is a diagonal operator, we deduce that the more general Formula~\eqref{eq:gradient_ground_state} holds. Taking the trace then yields
		\begin{equation*}
			h^d\Tr{\n{\Dh_{\xi_1}{\op}}^p} = \frac{2\,h^d}{(2\hbar)^{p/2}} \sum_{\n{\alpha}_1 = n}  \(\alpha_1+1\)^{p/2} = \frac{2\,h^d}{(2\hbar)^{p/2}} \sum_{k = 1}^{n+1} \binom{d+n-k-1}{d-2} k^{p/2}
		\end{equation*}
		and we deduce the result using the fact that $\binom{d+n}{d} = N = h^{-d}$ and $(n+1)^d \leq N\,d! \leq (n+d)^d$.
	\end{proof}

\section{The case of Schr\"odinger operators}

	Knowing Equation~\eqref{eq:soren_bound}, the core of the proof of Theorem~\ref{thm:regu} is nothing more than the use of an interpolation inequality for quantum Besov spaces.
	
	\begin{proof}[Proof of Theorem~\ref{thm:regu}]
		It follows from the Cwikel--Lieb--Rozenblum inequality \cite{cwikel_weak_1977} that 
		\begin{equation}\label{eq:CLR_ineq}
			\Nrm{\op}{\L^1} \leq C \Nrm{U_+}{L^{d/2}(\Rd)}
		\end{equation}
		where $U_+ = \max(U,0)$, and so in particular the above assumptions~\eqref{eq:conditions_u} imply that $\op\in\L^1$. Combined with Equation~\eqref{eq:soren_bound}, we deduce that $\op\in\cW^{1,1}$. Since $\op^2=\op$, it follows that $\Nrm{\op}{\L^\infty} \leq 1$. Now notice that if $\theta\in[0,1]$, $(s_0,s_1)\in[0,1]^2$ and $(p_0,p_1,q_0,q_1)\in[1,\infty]^4$, then it follows from the definition~\eqref{eq:Besov_def} of Besov norms and H\"older's inequality for the Lebesgue and Schatten norms that for any $s_\theta$, $p_\theta$ and $q_\theta$ such that
		\begin{equation*}
			s_\theta = \(1-\theta\)s_0 + \theta s_1, \qquad \frac{1}{p_\theta} = \frac{1-\theta}{p_0} + \frac{\theta}{p_1}, \qquad \frac{1}{q_\theta} = \frac{1-\theta}{q_0} + \frac{\theta}{q_1}
		\end{equation*}
		the following interpolation inequality holds 
		\begin{equation}\label{eq:interpolation_ineq_Besov}
			\Nrm{\op}{\dot{\cB}^{s_\theta}_{p_\theta,q_\theta}} \leq \Nrm{\op}{\dot{\cB}^{s_0}_{p_0,q_0}}^{1-\theta} \Nrm{\op}{\dot{\cB}^{s_1}_{p_1,q_1}}^\theta.
		\end{equation}
		We know from~\cite[Equation~(31)]{lafleche_quantum_2022} that for any $p\in[1,\infty]$,
		\begin{equation}\label{eq:comparison_besov_sobolev}
			\Nrm{\op}{\dot{\cB}^1_{p,\infty}} \leq 2 \Nrm{\op}{\dot{\cW}^{1,p}}.
		\end{equation}
		On the other hand, for any $p\in[1,\infty]$, the triangle inequality and the fact that the translation operators $\sfT_z$ are unitary, and so preserve the Schatten norms, yields
		\begin{equation*}
			\Nrm{\op}{\dot{\cB}^0_{\infty,\infty}} = \BNrm{\Nrm{\sfT_{2z}\op - 2\,\sfT_z\op + \op}{\L^\infty}}{L^\infty(\Rdd)} \leq 4 \Nrm{\op}{\L^\infty}.
		\end{equation*}
		Therefore, we deduce that
		\begin{equation*}
			\Nrm{\op}{\dot{\cB}^{1/p}_{p,\infty}} \leq \Nrm{\op}{\dot{\cB}^{0}_{\infty,\infty}}^{1-\theta} \Nrm{\op}{\dot{\cB}^{1}_{1,\infty}}^\theta \leq 2^{2-\theta} \Nrm{\op}{\L^\infty}^{1-\theta} \Nrm{\op}{\dot{\cW}^{1,1}}^\theta.
		\end{equation*}
		Together with the fact that $\Nrm{\op}{\L^p}$ is bounded uniformly in $\hbar$, it implies that $\Nrm{\op}{\cB^{1/p}_{p,\infty}}$ is bounded uniformly in $\hbar$, which finishes the proof of Equation~\eqref{eq:Besov_regu}. The interpolation Inequality~\eqref{eq:interpolation_ineq_Besov} implies now that $\Nrm{\op}{\cW^{s,2}} \leq C \Nrm{\op}{\cB^{1/p}_{p,\infty}}$. But by the fact that the Wigner transform is an isometry from $\L^2$ to $L^2(\Rdd)$ and the integral characterization of $H^s$, it is not difficult to see (see also \cite{lafleche_quantum_2022}) that $\Nrm{f_{\op}}{H^s} = C \Nrm{\op}{\cW^{s,2}}$, from which Equation~\eqref{eq:Sobolev_regu} follows.
	\end{proof}

\section{Projection operators converge to characteristic functions}

	In this section, we prove Theorem~\ref{thm:noregu}. We first introduce some tools that will be useful to our analysis. As in~\cite{werner_quantum_1984, lafleche_quantum_2022}, we define the semiclassical convolution as the weak integral
	\begin{equation}\label{eq:semiclassical_convolution_def}
		f \star \op = \op \star f := \intdd f(z)\,\sfT_z \op\d z.
	\end{equation}
	We will be interested by the convolution by the Gaussian function defined for any $z\in\Rdd$ by
	\begin{equation*}
		g_h(z) = \(2/h\)^d e^{-\n{z}^2/\hbar},
	\end{equation*}
	and we will use the notations
	\begin{equation*}
		\tilde{f} := g_h * f, \quad \text{ and } \quad \tildop := g_h \star \op.
	\end{equation*}
	In particular, the Husimi transform of an operator $\op$ is nothing but
	\begin{equation*}
		\tilde{f}_{\op} = g_h * f_{\op} = f_{\tildop}
	\end{equation*}
	while the Wick quantization of a measure $f$ (also sometimes called superposition of coherent sates, Töplitz operator or anti-Wick quantization) can be written
	\begin{equation*}
		\tildop_f = g_h\star \op_f = \op_{\tilde{f}} = f \star \op_{g_h}.
	\end{equation*}
	Young's inequality also holds for the semiclassical convolution (see e.g.~\cite{werner_quantum_1984,lafleche_quantum_2022}) and as a particular case we obtain the well known bounds for the Wick quantization and the Husimi transform
	\begin{equation}\label{eq:Toplitz_Husimi_bounds}
		\Nrm{\tildop_f}{\L^p} \leq \Nrm{f}{L^p(\Rdd)}, \quad \text{ and } \quad \Nrm{\tilde{f}_{\op}}{L^p(\Rdd)} \leq \Nrm{\op}{\L^p}.
	\end{equation}
	More generally, it is not difficult to deduce that the same inequality holds for quantum Sobolev and Besov spaces (see e.g. \cite[inequalities~(41) and~(47)]{lafleche_quantum_2022}), that is
	\begin{equation}\label{eq:Toplitz_Husimi_bounds_besov}
		\Nrm{\tildop_f}{\dot{\cB}^s_{p,q}} \leq \Nrm{f}{\dot{B}^s_{p,q}(\Rdd)}, \quad \text{ and } \quad \Nrm{\tilde{f}_{\op}}{\dot{B}^s_{p,q}(\Rdd)} \leq \Nrm{\op}{\dot{\cB}^s_{p,q}}.
	\end{equation}
	The main step in the proof of our main theorem is the following proposition.

	\begin{prop}\label{prop:cv_pure_states}
		Let $s\in(0,1]$, $p\in[2,\infty]$, $q\in[1,\infty]$ and $\op$ be (a sequence of) operators verifying $\op^2 = \op$, $h^d\Tr{\op}=1$ and $\op\in \dot{\cB}^s_{p,q}$ uniformly in $\hbar$. Then there exists a sequence $(\hbar_n)_{n\in\N}$ converging to $0$ and a sequence $(z_n)_{n\in\N}\in(\Rdd)^\N$ such that the sequence of translated Wigner transforms $\tilde{f}_{\op}(\cdot+z_n)$ for $\hbar = \hbar_n$ converges to $\mu : \Rdd\to \R$ different from $0$ and verifying
		\begin{equation*}
			\mu^2 = \mu \in B^s_{p,q}(\Rdd).
		\end{equation*}
	\end{prop}
	
	\begin{remark}
		The result is still true if there is no shifting sequence, i.e. if $z_n = 0$ for all $n\in\N$, except the fact that in this case it may happen that $\mu = 0$.
	\end{remark}

	\begin{proof}
		We take a discrete sequence of values for $\hbar$, $(\hbar_n)_{n\in\N}$, such that $\hbar_n\to 0$ when $n\to \infty$. We will then also take subsequences if necessary, but to simplify we will not write the $n$ dependency and just write $\hbar\to 0$. By assumption, there exists a constant $\dot{\cD}_s$ independent of $\hbar$ such that
		\begin{equation}\label{eq:Ds_bound}
			\Nrm{\op}{\dot{\cB}^s_{p,q}} \leq \dot{\cD}_s.
		\end{equation}
		By the quantum Sobolev inequalities \cite{lafleche_quantum_2022}, this implies that $\Nrm{\op}{\L^{p_s}} \leq \CS_{s,p}\dot{\cD}_s$ for some constant $\CS_{s,p}$ independent of $\hbar$ and with $\frac{1}{p_s} := \frac{1}{p} - \frac{s}{2d}$ if $p < \frac{2d}{s}$ and $p_s = \infty$ if $p > \frac{2d}{s}$ (and any $p_s\geq p$, that will be taken sufficiently large, if $p= \frac{2d}{s}$). By the fact that $\op^2=\op$ and by H\"older's inequality for Schatten norms, for any $r\geq 2$,
		\begin{equation*}
			1 = h^d\Tr{\op} \leq \Nrm{\op}{\L^1} = \Nrm{\op^2}{\L^1} \leq \Nrm{\op}{\L^2}^2 \leq \Nrm{\op}{\L^r}^{r'/2} \Nrm{\op}{\L^1}^{1-r'/2},
		\end{equation*}
		where $r'$ is the H\"older conjugate of $r$. Since $p_s\geq p\geq 2$, it implies in particular that
		\begin{equation}\label{eq:L1_bound}
			\Nrm{\op}{\L^1} \leq \Nrm{\op}{\L^{p_s}} \leq \CS_{s,p}\dot{\cD}_s.
		\end{equation}
		Combining the two above equations implies that for any $r\geq 2$,
		\begin{equation}\label{eq:lower_bound_Lr}
			\Nrm{\op}{\L^r} \geq (\CS_{s,p}\dot{\cD}_s)^{1-2/r'}.
		\end{equation}
		For the Husimi transform, by Formula~\eqref{eq:Toplitz_Husimi_bounds}, we deduce from Inequality~\eqref{eq:Ds_bound} and Inequality~\eqref{eq:L1_bound} the following bounds
		\begin{equation}\label{eq:pure_husimi_Lp}
			\Nrm{\tilde{f}_{\op}}{L^1(\Rdd)} \leq \CS_{s,p}\dot{\cD}_s \qquad \text{ and } \quad
			\Nrm{\tilde{f}_{\op}}{L^{p_s}(\Rdd)} \leq \CS_{s,p}\dot{\cD}_s.
		\end{equation}
		Moreover, since $p\in[1,p_s]$, by equations~\eqref{eq:Toplitz_Husimi_bounds} and~\eqref{eq:Toplitz_Husimi_bounds_besov} and H\"older's inequality, it holds
		\begin{equation}\label{eq:pure_husimi_regu}
			\Nrm{\tilde{f}_{\op}}{B^s_{p,q}(\Rdd)} \leq \Nrm{\op}{\L^p} + \Nrm{\op}{\dot{\cB}^s_{p,q}} \leq \(\CS_{s,p}+1\)\dot{\cD}_s =: \cD_s.
		\end{equation}
		These bound will allow us to extract weakly convergent subsequences. However, nothing prevents these sequences to converge weakly to $0$. Hence we will first use the ideas of the concentration-compactness principle~\cite{lions_concentration-compactness_1984, lions_concentration-compactness_1984-1} to prevent this. We mainly need to prevent the mass to escape at infinity, hence we follow the ideas of \cite[Lemma~I.1]{lions_concentration-compactness_1984-1} (see also \cite{lewin_describing_2010}) and look at all the sequences $T_z\tilde{f}_{\op} := \tilde{f}_{\op}(\cdot - z)$ and select the one with approximately the more mass in a given set. More precisely, by a diagonal argument, we can choose a sequence $z_\hbar$ such that
		\begin{equation*}
			\lim_{\hbar\to 0} \int_Q \n{T_{z_\hbar}\tilde{f}_{\op}} = \limsup_{\hbar\to 0} \(\sup_{z\in\Rdd} \int_Q \n{T_z\tilde{f}_{\op}}\) =: M_Q(\op)
		\end{equation*}
		where $Q$ is the unit cube $Q = \set{z\in\Rdd, \n{z}_\infty \leq 1/2}$. This new sequence of functions $T_{z_\hbar}\tilde{f}_{\op}$ verifies the same bounds~\eqref{eq:pure_husimi_Lp} and~\eqref{eq:pure_husimi_regu} as $\tilde{f}_{\op}$.
		Hence $T_{z_\hbar}\tilde{f}_{\op}\in B^s_{p,q}\cap L^1 \cap L^{p_s}$ uniformly in $\hbar$ and so, up to a subsequence, converges weakly in $B^s_{p,q}(\Rdd)$ and strongly in $L^p_{\loc}(\Rdd)$ to a function $\mu \in B^s_{p,q}\cap L^1 \cap L^{p_s}$ verifying 
		\begin{equation*}
			\Nrm{\mu}{B^s_{p,q}(\Rdd)} \leq \cD_s \quad\text{ and }\quad \Nrm{\mu}{L^1(Q)} = M_Q(\op).
		\end{equation*}
		Let $\sigma \in (0,s)$ and $\frac{1}{p_\sigma} := \frac{1}{p} - \frac{\sigma}{2d}$. By the Gagliardo--Nirenberg interpolation inequalities, for any $r\in(1,p_\sigma)$ and $\theta\in(0,1)$ such that $\frac{1}{r} = \frac{\theta}{p_\sigma} + 1-\theta$, there exists a constant $C>0$ independent of $\hbar$ and $\op$ such that
		\begin{equation*}
			\intdd \n{\tilde{f}_{\op}}^r = \sum_{k\in\Z^d} \int_Q \n{T_k \tilde{f}_{\op}}^r \leq C \sum_{k\in\Z^d} \(\int_Q \n{T_k \tilde{f}_{\op}}\)^{\(1-\theta\)r} \Nrm{\tilde{f}_{\op}}{W^{\sigma,p}(Q)}^{\theta\,r}.
		\end{equation*}
		In particular, one can choose $r$ verifying $\theta\,r = p$ by taking $r = p\(1+\frac{\sigma}{2d}\)$. Hence, it yields
		\begin{equation*}
			\intdd \n{\tilde{f}_{\op}}^r \leq C \sup_{k\in\Z^d} \(\int_Q \n{T_k \tilde{f}_{\op}}\)^{\(1-\theta\)r} \Nrm{\tilde{f}_{\op}}{W^{\sigma,p}(\Rdd)}^p
		\end{equation*}
		and so since $B^s_{p,q} \subset W^{\sigma,p}$, it follows that
		\begin{equation*}
			\Nrm{\ttildop}{\L^r} \leq \Nrm{\tilde{f}_{\op}}{L^r(\Rdd)} \leq C \(\sup_{z\in\Z^d}\int_Q \n{T_z \tilde{f}_{\op}}\)^\frac{\sigma}{2d+\sigma} \cD_s^\frac{\sigma}{2d+\sigma}
		\end{equation*}
		where we used Inequality~\eqref{eq:Toplitz_Husimi_bounds} with $\ttildop = \tildop_{\tilde{f}_{\op}}$ to get the first inequality. On the other hand, we know from \cite[Inequality~(42)]{lafleche_quantum_2022} that $\sNrm{\op-\ttildop}{\L^p} \leq C\,\hbar^s\Nrm{\op}{\cB^s_{p,q}}$. Since $r\geq p$, taking into account the $h$ appearing in the definition of the $\L^p$ norms, the inclusions between Schatten norms yields $\sNrm{\op-\ttildop}{\L^r} \leq C\,\hbar^{s-2d\(\frac{1}{p}-\frac{1}{r}\)}\Nrm{\op}{\cB^s_{p,q}}$, and so 
		\begin{equation*}
			\Nrm{\op}{\L^r} - C\,\hbar^{s-\frac{2d\sigma}{\(2d+s\)p}} \,\cD_s \leq C \(\sup_{z\in\Z^d}\int_Q \n{T_z \tilde{f}_{\op}}\)^\frac{\sigma}{2d+\sigma} \cD_s^\frac{\sigma}{2d+\sigma}.
		\end{equation*}
		Noticing that the exponent of $\hbar$ appearing in the above equation is positive and that since $r\geq p \geq 2$ we can use Inequality~\eqref{eq:lower_bound_Lr}, we can take $\hbar \to 0$ to get
		\begin{equation*}
			(\CS_{s,p}\dot{\cD}_s)^{1-2/r'} \leq C\, \cD_s^\frac{\sigma}{2d+\sigma} \Nrm{\mu}{L^1(Q)}^\frac{\sigma}{2d+\sigma}.
		\end{equation*}
		In particular, $\mu \neq 0$.
		
		Now it remains to prove that $\mu^2$ is also limit of the sequence $T_{z_\hbar}\tilde{f}_{\op}$. For that, we will use the following inequality proved in \cite[Lemma~3.1]{chong_l2_2023}. For any $p\geq 2$, it holds
		\begin{equation*}
			\Nrm{\frac{\tildop_f\tildop_g+\tildop_f\tildop_g}{2}-\tildop_{fg}}{\L^{p/2}} \leq 2^{d+1}\, d\,\hbar \Nrm{\nabla f}{L^p(\Rdd)}\Nrm{\nabla g}{L^p(\Rdd)},
		\end{equation*}
		The above inequality, together with the fact that by H\"older's inequality
		\begin{equation*}
			\Nrm{\frac{\tildop_f\tildop_g+\tildop_f\tildop_g}{2}-\tildop_{fg}}{\L^{p/2}} \leq  \Nrm{\tildop_f}{\L^p}\Nrm{\tildop_g}{\L^p} + \Nrm{fg}{L^{p/2}(\Rdd)} \leq 2\Nrm{f}{L^p(\Rdd)}\Nrm{g}{L^p(\Rdd)}
		\end{equation*}
		leads by bilinear real interpolation (see e.g.~\cite[Lemma~28.3]{tartar_introduction_2007}) and taking $f = g = \mu$ to
		\begin{equation*}
			\Nrm{\tildop_\mu^2-\tildop_{\mu^2}}{\L^{p/2}} \leq C\,\hbar^\sigma \Nrm{\mu}{B^\sigma_{p,\tilde{q}}(\Rdd)}^2 \leq C\,\hbar^\sigma \Nrm{\mu}{B^s_{p,q}(\Rdd)}^2,
		\end{equation*}
		where $\tildop_\mu^2 = (\tildop_\mu)^2$, $\tilde{q} \geq 2$ and we used the fact that $(L^p, W^{1,p})_{s,\tilde{q}} = B^{s}_{p,\tilde{q}}$ (see e.g. \cite[Section~2.4.2]{triebel_interpolation_1978}) and the continuous embedding $B^s_{p,q} \subset B^\sigma_{p,\tilde{q}}$ for any $\sigma < s$. That is, squaring an operator is an operation close to squaring a function. By the properties of the Husimi transform, we deduce that
		\begin{equation*}
			\Nrm{\tilde{f}_{\tildop_\mu^2} - \tilde{f}_{\tildop_{\mu^2}}}{L^r(\Rdd)} \leq C\,\hbar^\sigma \,\cD_s^2.
		\end{equation*}
		Noticing that $\tilde{f}_{\tildop_{\mu^2}} = \tilde{\tilde{f}}_{\op_{\mu^2}} = g_{2h} * f_{\op_{\mu^2}} = g_{2h} * \mu^2$, we deduce that $g_h * f_{\tildop_\mu^2} - g_h * \mu^2$ converges to $0$ in $L^r(\Rdd)$. On another side, since $g_h$ is an approximation of the identity, $g_h*\mu^2 - \mu^2$ also converges to $0$ in $L^r(\Rdd)$. Therefore, strongly in $L^r(\Rdd)$, the following convergence holds true
		\begin{equation}\label{eq:strong_cv}
			\tilde{f}_{\tildop_\mu^2} \underset{h \to 0}{\longrightarrow} \mu^2.
		\end{equation}
		Now we want to prove that $\tilde{f}_{\tildop_\mu^2}$ is close to $\mu$. To this end we use the fact that $\op_h = \op_h^2$ where $\op_h = \sfT_{z_h}\op$, and write
		\begin{equation}\label{eq:triangle_equality}
			\mu - \tilde{f}_{\tildop_\mu^2} = \(\mu - \tilde{f}_{\op_h}\) + \(\tilde{f}_{\op_h^2} - \tilde{f}_{\ttildop_h^2}\) + \(\tilde{f}_{\ttildop_h^2} - \tilde{f}_{\tildop_\mu^2}\).
		\end{equation}
		We already know that the first term on the right-hand-side converges to $0$ in $L^p_{\loc}(\Rdd)$.  To bound the second term, we use the Husimi transform bound~\eqref{eq:Toplitz_Husimi_bounds} to write
		\begin{equation*}
			\Nrm{\tilde{f}_{\op_h^2} - \tilde{f}_{\ttildop_h^2}}{L^r(\Rdd)} \leq \Nrm{\op^2 - \ttildop_h^2}{\L^r}.
		\end{equation*}
		To control this term, it is useful to remove the squares using the fact that
		\begin{equation}\label{eq:dif_square_identity}
			\op_h^2 - \ttildop_h^2 = \frac{1}{2} \(\(\op_h-\ttildop_h\)\(\op_h+\ttildop_h\) + \(\op_h+\ttildop_h\)\(\op_h-\ttildop_h\)\).
		\end{equation}
		By the semiclassical Young convolution inequality (see \cite{werner_quantum_1984, lafleche_quantum_2022}) and the fact that $\op_h^2 = \op_h$, it holds $\Nrm{\ttildop_h}{\L^p} \leq \Nrm{\op_h}{\L^p} \leq 1$. Therefore, it follows from~\eqref{eq:dif_square_identity} and H\"older's inequality that
		\begin{equation*}
			\Nrm{\op_h^2 - \ttildop_h^2}{\L^r} \leq \(\Nrm{\op_h}{\L^p} + \Nrm{\ttildop_h}{\L^p}\) \Nrm{\op_h-\ttildop_h}{\L^p} \leq 2 \Nrm{\op_h-\ttildop_h}{\L^p}.
		\end{equation*}
		Since by \cite{lafleche_quantum_2022}, we know that $\Nrm{\op_h-\ttildop_h}{\L^p} \leq C\,\hbar^s \Nrm{\op_h}{\cB^s_{p,q}}$, it yields
		\begin{equation}\label{eq:double_convolution_error}
			\Nrm{\op_h^2 - \ttildop_h^2}{\L^r} \leq C\,\hbar^s \,\cD_s.
		\end{equation}
		It remains to treat the last term on the right-hand-side of Identity~\eqref{eq:triangle_equality}. The idea is similar to the second term, but since there is only convergence of $\tilde{f}_{\op}$ to $\mu$ weakly or locally, we pass to the weak topology. Take $\varphi\in L^{r'}(\Rdd)$ a test function. Then using again Identity~\eqref{eq:dif_square_identity} with $\op_h$ replaced by $\tildop_\mu$ yields
		\begin{multline*}
			\intdd \(\tilde{f}_{\ttildop_h^2} - \tilde{f}_{\tildop_\mu^2}\)\varphi = \intdd \(f_{\ttildop_h^2} - f_{\tildop_\mu^2}\) \tilde{\varphi} = \tfrac{h^d}{2} \Tr{\(\tildop_{\mu}-\ttildop_h\)\opnu + \opnu^*\(\tildop_{\mu}-\ttildop_h\)}
			\\
			= \Re{\intdd \(\tilde\mu - \tilde{\tilde{f}}_{\op_h}\) f_{\opnu}} = \Re{\intdd \(\mu - \tilde{f}_{\op_h}\) \tilde{f}_{\opnu}}
		\end{multline*}
		where $\opnu = \(\tildop_{\mu}+\ttildop_h\)\tildop_{\varphi}$, and $\tilde{f}_{\opnu}$ is bounded in $L^{p'}(\Rdd)$ uniformly in $\hbar$ since by H\"older's inequality with $\frac{1}{p'} = \frac{1}{p} + \frac{1}{r'}$
		\begin{equation*}
			\Nrm{\tilde{f}_{\opnu}}{L^{p'}(\Rdd)} \leq \Nrm{\(\tildop_{\mu}+\ttildop_h\)\tildop_{\varphi}}{\L^{p'}} \leq \(\Nrm{\mu}{L^p(\Rdd)}+\Nrm{\op_h}{\L^p}\) \Nrm{\tildop_{\varphi}}{\L^{r'}} \leq 2\Nrm{\varphi}{L^{r'}(\Rdd)}.
		\end{equation*}
		Taking an approximation of $f_{\opnu}$ by compactly supported functions and using the fact that $\mu- \tilde{f}_{\op_h}$ converges strongly to $0$ in $L^p_{\loc}(\Rdd)$ leads to the fact that $\tilde{f}_{\ttildop_h^2} - \tilde{f}_{\tildop_\mu^2}$ converges weakly to $0$ in $L^r(\Rdd)$. Recalling Inequality~\eqref{eq:double_convolution_error} and coming back to Equation~\eqref{eq:triangle_equality}, we proved that
		\begin{equation*}
			\tilde{f}_{\tildop_\mu^2} \underset{h\to 0}{\rightharpoonup} \mu^2
		\end{equation*}
		weakly in $L^r(\Rdd)$. Together with Equation~\eqref{eq:strong_cv}, this proves that $\mu = \mu^2$.
	\end{proof}
	
	\begin{proof}[Proof of Theorem~\ref{thm:noregu}]
		Assume $s > 1/p$, or $s= 1/p$ and $q < \infty$, and let $\op = (\op_\hbar)_{\hbar\in(0,1)}$ be a sequence of operators verifying the assumption of the theorem and such that the norm $\Nrm{\op_\hbar}{\dot{\cB}^s_{p,q}}$ does not converge to $\infty$. Then there exists a subsequence of $\op$ bounded in $\cB^s_{p,q}$ uniformly in $\hbar$. We now write $\op = (\op_\hbar)_{\hbar\in(0,1)}$ this subsequence.
	
		\step{Case 1} Assume first that $p\in [2,\infty]$. Then by Proposition~\ref{prop:cv_pure_states}, there exists a function $\mu \in B^s_{p,q}(\Rdd)$ such that $\mu^2 = \mu \geq 0$, that is $\mu$ is the characteristic function of some set. But this is known to be false: the characteristic function of a set cannot have such regularity, see for example~\cite{sickel_regularity_2021}. This proves the result in the case $p\geq 2$.
		
		\step{Case 2} Now assume that $p\in[1,2)$ and let $\theta= \frac{1}{2s}$, $r = \frac{q}{\theta}$ and $p_0 = \frac{2s-1}{sp-1} \,p$. Since $p < 2$, we deduce that $s \in (1/2,1)$ and so $\theta\in(1/2,1)$ and $r\in(q,\infty)$, and since $p\in[1,2)$ and $s\in(1/2,1)$ and $sp\geq 1$, we get that $p_0 \in (p,\infty]$ (with the convention that $p_0 = \infty$ if $sp=1$). One can rewrite the definition of $p_0$ and $r$ as follows
		\begin{equation*}
			\frac{1}{2} = \frac{\theta}{p} + \frac{1-\theta}{p_0}, \qquad \frac{1}{r} = \frac{\theta}{q}.
		\end{equation*}
		Hence it follows from Inequality~\eqref{eq:interpolation_ineq_Besov} that
		\begin{equation*}
			\Nrm{\op}{\dot{\cB}^{1/2}_{2,r}} \leq  \Nrm{\op}{\dot{\cB}^{s}_{p,q}}^\theta \(4\Nrm{\op}{\L^{p_0}}\)^{1-\theta}.
		\end{equation*}
		If $\op$ is a self-adjoint operator, since $\op^2 = \op$, $\Nrm{\op}{\L^p} = \Nrm{\op}{\L^1}^{1/p}$ is independent of $\hbar$ by assumption. Hence if $\Nrm{\op}{\dot{\cB}^{s}_{p,q}}$ is bounded uniformly in $\hbar$, so will be $\Nrm{\op}{\dot{\cB}^{1/2}_{2,r}}$, contradicting the case $p=2$ already proved in the first part of this proof. This finishes the proof of Formula~\eqref{eq:Besov_noregu}. In the case when $s\in(0,1)$, then Equation~\eqref{eq:Sobolev_noregu} is an immediate consequence of Formula~\eqref{eq:Besov_noregu} and the fact that $\gamma_{s,p}^{1/p} \Nrm{\op}{\dot{\cB}^s_{p,p}} \leq 2 \Nrm{\op}{\dot{\cW}^{s,p}}$ by the triangle inequality. In the case when $s=1$, then Formula~\eqref{eq:Sobolev_noregu} follows from Formula~\eqref{eq:Besov_noregu} and Inequality~\eqref{eq:comparison_besov_sobolev}.
		
		In the case when $\op\in\dot{\cB}^s_{p,q}$ is not a self-adjoint operator but we know that $\op\in\L^{2+\eps}$ uniformly in $\hbar$ for some $\eps>0$, then by the interpolation inequality~\eqref{eq:interpolation_ineq_Besov}, we deduce that $\op\in\dot{\cB}^\alpha_{\beta,\gamma}$ for some $\alpha>0$ and some $\beta>2$. Hence by  Proposition~\ref{prop:cv_pure_states}, up to a shifting sequence and a subsequence, the Husimi transform of $\op$ converges to some indicator function $0 \neq \mu = \mu^2 \in B^\alpha_{\beta,\gamma}(\Rdd)$. This regularity is not forbidden for characteristic functions if $\alpha$ is small enough. However, since the Husimi transform of $\op$ is also in $B^s_{p,q}(\Rdd)$ uniformly in $\hbar$, one also deduces $\mu\in B^s_{p,q}(\Rdd)$, which is not possible for characteristic functions. This proves the theorem in this case.
		
		In the particular case $p \leq \frac{2\,d}{d+s}$, then the fact that $\op\in\L^{2+\eps}$ uniformly in $\hbar$ follows directly from the fact that $\op\in\dot{\cB}^s_{p,q}$ uniformly in $\hbar$ by the quantum Sobolev inequalities~\cite{lafleche_quantum_2022}, and so the proof follows as in the previous paragraph.
	\end{proof}

\section{Proof of the Weyl law in Sobolev spaces}

	\begin{proof}[Proof of Theorem~\ref{thm:regu_noregu}]
		The beginning of the proof is classical (see e.g.~\cite{fournais_semi-classical_2018}). Notice first that
		\begin{equation}\label{eq:kinetic_energy_bound}
			h^d \Tr{\op \n{\opp}^2} = h^d \Tr{\op \(\n{\opp}^2-U(x)\)} + h^d \Tr{\op \,U(x)} \leq M_0\,\Nrm{U_+}{L^\infty(\Rd)}
		\end{equation}
		where $M_h = h^d\Tr{\op}$, from which it follows that $h^d \Tr{\op \n{\opp}^2}$ is bounded uniformly in $\hbar$. Since $0\leq \op\leq 1$ and $\op$ is bounded in $\L^1$ uniformly in $\hbar$ by Inequality~\eqref{eq:CLR_ineq}, up to a subsequence, the Husimi transform of $\op$ converges weakly (for example in $L^p(\Rdd)$ for $p\in(1,\infty)$) to some function $f$ verifying $0\leq f\leq 1$. The Wigner transform of $\op$ then converges weakly to the same limit (see \cite{lions_sur_1993}). On the other hand, the classical asymptotic formula for the eigenvalue counting function of the Schr\"odinger operator \cite{martin_bound_1972, tamura_asymptotic_1974, simon_analysis_1976, cwikel_weak_1977} yields
		\begin{equation*}
			\intdd f_{\op}(x,\xi) \d x\d \xi = h^d\Tr{\op} \underset{\hbar\to 0}{\to} \intdd \indic_{\n{\xi}^2\leq U(x)}(x,\xi) \d x\d \xi = M_0
		\end{equation*}
		while from the Weyl law for the energy of the Schrödinger operator (see e.g. \cite{lieb_inequalities_1976, hundertmark_new_2000}), it holds
		\begin{multline*}
			\intdd f_{\op} \(U(x) - \n{\xi}^2\) = h^d \Tr{\op \(U(x) - \n{\opp}^2\)}  = h^d \Tr{\(U(x) - \n{\opp}^2\)_+}
			\\
			\underset{\hbar\to 0}{\to} \iintd \(U(x) - \n{\xi}^2\)_+ \d x\d \xi = \max_{g \in G} \iintd \(U(x) - \n{\xi}^2\) g(x,\xi) \d x\d \xi,
		\end{multline*}
		where $G = \set{g:\Rdd\to \R | 0\leq g\leq 1 \text{ and } \intdd g = M_0}$. Hence, by the bathtub principle, we deduce that $f = \indic_{\n{\xi}^2 \leq U(x)}$.
		
		If $\op$ is bounded in $\cB^s_{p,q}$ uniformly in $\hbar$, then by Equation~\eqref{eq:Toplitz_Husimi_bounds_besov}, its Husimi transform is bounded in $B^s_{p,q}(\Rdd)$ uniformly in $\hbar$ and so $f$ will be in $B^s_{p,q}(\Rdd)$ as well. If $\sqrt{U}\in C^\alpha$ with with $\alpha\in(0,1)$, then $\indic_{\n{\xi}^2\leq U(x)}\in B^{\alpha/p}_{p,\infty}(\Rdd)$. However, there exist examples of functions $U\in C^\alpha$ such that $\indic_{\n{\xi}^2\leq U(x)} \notin B^s_{p,q}(\Rdd)$ for all $s > \alpha/p$ (see e.g. \cite[Section~4.4]{ sickel_regularity_2021} or \cite{triebel_fractals_2010}). For these functions, we therefore deduce that $\op$ can never be bounded in $\cB^s_{p,q}$ uniformly in $\hbar$, proving Equation~ \eqref{eq:existence_of_irregularity}.
		
		\step{Step 2: convergence of the Husimi transform} In the particular case when $U$ verifies Assumption~\eqref{eq:conditions_u} with $\Omega$ bounded, then by Theorem~\ref{thm:regu} and Equation~\eqref{eq:Toplitz_Husimi_bounds_besov}, we deduce that the Husimi transform of $\op$ is bounded uniformly in $\hbar$ in $B^{1/p}_{p,\infty}(\Rdd)$ and so converges weakly in $B^{1/p}_{p,\infty}(\Rdd)$ to~$f$.
		
		To get strong convergence in $L^p$, it remains to use that the moments of $\op$ are bounded. Indeed, the Husimi transform of $\op$ verifies $0\leq \tilde{f}_{\op} \leq 1$ and
		\begin{equation*}
			\Nrm{\tilde{f}_{\op} \n{\xi}}{L^2(\Rdd)}^2 \leq \intdd g_h*f_{\op} \n{\xi}^2\d x\d\xi = h^d\Tr{\op\n{\opp}^2} + d\,\hbar,
		\end{equation*}
		which is bounded uniformly in $\hbar\in(0,1]$ by Inequality~\eqref{eq:kinetic_energy_bound}. On the other hand, it follows from Agmon estimates \cite{agmon_lectures_1982, agmon_bounds_1985} that there exist positive constants $c$ and $C$ independent of $\hbar$ and $N=h^{-d}$ such that for all eigenvalues $(\psi_j)_{j\in\set{1,\dots,N}}$ of $\op$,
		\begin{equation*}
			\Nrm{e^{c\n{x}}\,\psi_j(x)}{L^2(\Rd)} \leq C_\eps,
		\end{equation*}
		where $C_\eps$ depends on the $\eps$ appearing in Assumption \eqref{eq:conditions_u}. Since
		\begin{equation*}
			\op = \sum_{j=1}^N \ket{\psi_j}\bra{\psi_j},
		\end{equation*}
		where by Inequality~\eqref{eq:CLR_ineq}, $Nh^d\leq C \Nrm{U_+}{L^{d/2}(\Rd)}$, this implies that
		\begin{equation*}
			\Nrm{\tilde{f}_{\op} \n{x}}{L^2(\Rdd)}^2 -d\hbar \leq h^d \Tr{\op \n{x}^2} = h^d\sum_{j=1}^N \Nrm{\n{x}\psi_j(x)}{L^2(\Rd)}^2 \leq C_\eps \Nrm{U_+}{L^{d/2}(\Rd)}.
		\end{equation*}
		From these $\n{\xi}$ and $\n{x}$ moments bounds and the uniform boundedness of $\tilde{f}_{\op}$ in $H^s(\Rdd)$ for $s<1/2$, it follows that $\tilde{f}_{\op}$ converges strongly to $f$ in $L^2(\Rdd)$ by the classical Riesz--Fréchet--Kolmogorov criterion of compactness in $L^p$ (see e.g. \cite[Theorem~IV.25]{brezis_analyse_2005}). By interpolation and the uniform bound of $\tilde{f}_{\op}$ in $L^1(\Rdd)$ and $L^\infty(\Rdd)$, this implies the strong convergence in $L^p$ for any $p\in(1,\infty)$. More precisely, for $p\in[2,\infty)$,
		\begin{equation*}
			\Nrm{\tilde{f}_{\op}-f}{L^p(\Rdd)} \leq \Nrm{\tilde{f}_{\op}-f}{L^2(\Rdd)}^{2/p} \Nrm{\tilde{f}_{\op}-f}{L^\infty(\Rdd)}^{1-2/p} \leq \Nrm{\tilde{f}_{\op}-f}{L^2(\Rdd)}^{2/p}
		\end{equation*}
		since $-1\leq \tilde{f}_{\op}-f \leq 1$, and similarly for $p\in(1,2)$ by replacing $L^\infty$ by $L^1$. Strong convergence also holds in $L^1(\Rdd)$ following the same scheme of proof as for $L^2(\Rdd)$, but using the boundedness of $\Nrm{\tilde{f}_{\op}}{W^{1,1}(\Rdd)}$, $\Nrm{\tilde{f}_{\op}\n{x}^2}{L^1(\Rdd)}$ and $\Nrm{\tilde{f}_{\op}\n{\xi}^2}{L^1(\Rdd)}$ uniformly in $\hbar$.
		
		Convergence in $W^{s,p}(\Rdd)$ for all $s\in (0,1/p)$ follows similarly from the fact that these are interpolation spaces between $L^p(\Rdd)$ and $B^{1/p}_{p,\infty}(\Rdd)$, proving Equation~\eqref{eq:CV_Sobolev_Husimi}.
		
		\step{Step 3: convergence of the Wigner transform} For the Wigner transform $f_{\op}$, convergence in $L^2(\Rdd)$ follows from the fact that for any $s<1/2$,
		\begin{equation*}
			\Nrm{f_{\op} - \tilde{f}_{\op}}{L^2(\Rdd)} \leq C\,h^s \Nrm{f_{\op}}{H^s(\Rdd)},
		\end{equation*}
		which converges to $0$ when $h\to 0$ by Equation~\eqref{eq:Sobolev_regu}. Hence the convergence of the Husimi transform implies the convergence of the Wigner transform. For $p\geq 2$, one can use the classical Sobolev embeddings $H^{s}(\Rdd) \to W^{\sigma,p}(\Rdd)$ for $\frac{1}{2}-\frac{s-\sigma}{2d} \leq \frac{1}{p} \leq \frac{1}{2}$ to get that the Wigner transforms are bounded uniformly in $\hbar$ in $W^{\sigma,p}(\Rdd)$ for $2\leq p < \frac{4d}{2\(d+\sigma\)-1}$, and the strong convergence in $W^{t,p}(\Rdd)$ for $t\in[0,\sigma)$ follows again by an interpolation argument, proving Equation~\eqref{eq:CV_Sobolev_Wigner}.
		
		Finally notice that the boundedness of moments in $x$ and $\opp$ of the operator $\op$ allows to consider lower values of $p$. Indeed, with the notation $X = x+y/2$ and $Y= x-y/2$, so that $x=X+Y$, the definition of the Wigner transform together with the multinomial formula yield for any $n\in2\N$,
		\begin{align*}
			f_{\op} \n{x}^n &= \intd e^{-i\,y\cdot\xi/\hbar} \n{X+Y}^n \op(X,Y)\d y
			\\
			&= \intd e^{-i\,y\cdot\xi/\hbar} \(\n{X}^2+\n{Y}^2+2\,X\cdot Y\)^{n/2} \op(X,Y)\d y
			\\
			&= \sum_{\n{\alpha+\beta}_1= n} C_{\alpha,\beta} \intd e^{-i\,y\cdot\xi/\hbar} X^\alpha\, \op(X,Y) \,Y^{\beta} \d y = \sum_{\n{\alpha+\beta}_1= n} C_{\alpha,\beta}\, f_{x^\alpha\op\,x^\beta}
		\end{align*}
		where the sum is taken over all the multiindices $(\alpha,\beta)\in \N^d\times \N^d$ such that
		\begin{equation*}
			\n{\alpha+\beta}_1 := \sum_{j=1}^d \alpha_j + \beta_j = n
		\end{equation*}
		and where $x^\alpha$ means $x_1^{\alpha_1}\dots x_d^{\alpha_d}$. The constants $C_{\alpha,\beta}$ verify that their sum is $\(2\sqrt d\)^n$, as can be seen by taking $X=Y=(1,\dots,1)$. Hence, taking the $L^2(\Rdd)$ norm and using the triangle inequality and the fact that the Wigner trnasform is an isometry from $\L^2$ to $L^2(\Rdd)$ yields
		\begin{equation*}
			\Nrm{f_{\op} \n{x}^n}{L^2(\Rdd)} \leq \sum_{\n{\alpha+\beta}_1= n} C_{\alpha,\beta}\, \Nrm{x^\alpha\op\,x^\beta}{\L^2} \leq \(2\sqrt d\)^n \sup_{\n{\alpha+\beta}_1=n} \Nrm{x^\alpha\op\,x^\beta}{\L^2}
		\end{equation*}
		By the boundedness of the operator of multiplication by $\frac{x^\alpha}{\n{x}^{\n{\alpha}_1}}$,
		\begin{equation*}
			\Nrm{x^\alpha\op\,x^\beta}{\L^2} \leq \Nrm{\n{x}^{\n{\alpha}_1}\op \n{x}^{\n{\beta}_1}}{\L^2} \leq \Nrm{\op\n{x}^{\n{\alpha+\beta}_1}}{\L^2}
		\end{equation*}
		where the last inequality follows from \cite[Inequality~(56)]{lafleche_strong_2023} (which follows mainly from H\"older's inequality for Schatten norms and the Araki--Lieb--Thirring inequality~\cite{araki_inequality_1990}). Hence
		\begin{equation*}
			\Nrm{f_{\op} \n{x}^n}{L^2(\Rdd)}^2 \leq \(4d\)^n \Nrm{\op\n{x}^n}{\L^2}^2 = \(4d\)^n h^d\Tr{\op \n{x}^{2n}}
		\end{equation*}
		where we used the that fact that $\op^2=\op$. This last expression is bounded uniformly with respect to $\hbar$ by Agmon estimates. A similar proof replacing $x$ by $\xi$ yields
		\begin{equation*}
			\Nrm{f_{\op} \n{\xi}^n}{L^2(\Rdd)}^2 \leq \(4d\)^nh^d\Tr{\op \n{\opp}^{2n}} = \(4d\)^n h^d \sum_{j=1}^N \Inprod{\psi_j}{\n{\opp}^{2n}\psi_j}
		\end{equation*}
		where the $\psi_j$ verify $\n{\opp}^2 \psi_j = \(U(x)+\lambda_j\)\psi_j$ for some $\lambda_j\leq 0$. In particular, if $n$ is odd, then $\(U(x)-\mu_j\)^n \leq U_+^n(x)$, so that
		\begin{equation*}
			\Inprod{\psi_j}{\n{\opp}^{2n}\psi_j} \leq \Inprod{\psi_j}{U_+^n\psi_j} \leq \Nrm{U_+}{L^\infty(\Rd)}^n
		\end{equation*}
		which for the Wigner transform yields
		\begin{align*}
			\Nrm{f_{\op} \n{\xi}^n}{L^2(\Rdd)}^2 &\leq \(4d\)^nh^d\Tr{\op \n{\opp}^{2n}} \leq \(4d\)^n h^d \sum_{j=1}^N \Nrm{U_+}{L^\infty(\Rd)}^n
			\\
			&\leq C_d \(4d\)^n \Nrm{U_+}{L^{d/2}(\Rd)} \Nrm{U_+}{L^\infty(\Rd)}^n.
		\end{align*}
		From these weighted $L^2$ estimates, $L^p$ estimates for $p<2$ follow by H\"older's inequality since whenever $n > 2d$
		\begin{equation}\label{eq:moments_L2_to_L1}
			\Nrm{f_{\op}}{L^1(\Rdd)} \leq C_{d,n} \Nrm{f_{\op} \(1+\n{x}^n+\n{\xi}^n\)}{L^2(\Rdd)}
		\end{equation}
		with $C_{d,n}$ independent of $\hbar$. The boundedness in $L^1(\Rdd)$ and in $H^s(\Rdd)$ for $s<1/2$ together with the convergence in $L^2(\Rdd)$ imply by interpolation the convergence in $W^{s,p}(\Rdd)$ for any $s < 1/p'$ with $p\in(1,2)$.
		
		Notice that similarly as for the above moments bounds~\eqref{eq:moments_L2_to_L1}, one has more generally whenever $n > 2d + k$
		\begin{equation}\label{eq:moments_L1}
			\Nrm{f_{\op}\(1+\n{x}^k+\n{\xi}^k\)}{L^1(\Rdd)} \leq C_{d,n} \Nrm{f_{\op} \(1+\n{x}^n+\n{\xi}^n\)}{L^2(\Rdd)}
		\end{equation}
		Together with the convergence in $L^2(\Rdd)$ of the Wigner transform, which implies the convergence in $L^1$ on any compact of $\Rdd$, this gives the convergence in $L^1(\Rdd)$.
	\end{proof}

\medskip
\paragraph{\bf Acknowledgment.} This project has received funding from the European Research Council (ERC) under the European Union’s Horizon 2020 research and innovation program (grant agreement No 865711).


\renewcommand{\bibname}{\centerline{Bibliography}}
\bibliographystyle{abbrv} 
\bibliography{Vlasov}

\begin{thebibliography}{10}

\bibitem{agmon_lectures_1982}
S.~Agmon.
\newblock {\em Lectures on {{Exponential Decay}} of {{Solutions}} of
  {{Second-Order Elliptic Equations}}: {{Bounds}} on {{Eigenfunctions}} of
  {{N-Body Schrodinger Operations}}}.
\newblock Number~29 in Princeton {{Legacy Library}}, {{Mathematical Notes}}.
  {Princeton University Press}, 1982.

\bibitem{agmon_bounds_1985}
S.~Agmon.
\newblock Bounds on exponential decay of eigenfunctions of {{Schr\"odinger}}
  operators.
\newblock In {\em Schr\"odinger {{Operators}}}, Lecture {{Notes}} in
  {{Mathematics}}, pages 1--38, {Berlin, Heidelberg}, 1985. {Springer}.

\bibitem{araki_inequality_1990}
H.~Araki.
\newblock On an inequality of {{Lieb}} and {{Thirring}}.
\newblock {\em Letters in Mathematical Physics}, 19(2):167--170, 1990.

\bibitem{benedikter_effective_2022}
N.~Benedikter.
\newblock Effective {{Dynamics}} of {{Interacting Fermions}} from
  {{Semiclassical Theory}} to the {{Random Phase Approximation}}.
\newblock {\em Journal of Mathematical Physics}, 63(8):081101, Aug. 2022.

\bibitem{brezis_analyse_2005}
H.~Brezis.
\newblock {\em {Analyse fonctionnelle : Th\'eorie et applications}}.
\newblock {Sciences Sup}. {Dunod}, {Paris}, nouvelle pr\'esentation 2005
  edition, May 2005.

\bibitem{chong_many-body_2021}
J.~J. Chong, L.~Lafleche, and C.~Saffirio.
\newblock From {{Many-Body Quantum Dynamics}} to the
  {{Hartree}}\textendash{{Fock}} and {{Vlasov Equations}} with {{Singular
  Potentials}}.
\newblock {\em arXiv:2103.10946}, pages 1--74, Mar. 2021.

\bibitem{chong_l2_2023}
J.~J. Chong, L.~Lafleche, and C.~Saffirio.
\newblock On the {$L^2$} {{Rate}} of {{Convergence}} in the {{Limit}} from the
  {{Hartree}} to the {{Vlasov}}\textendash{{Poisson Equation}}.
\newblock {\em Journal de l'\'Ecole polytechnique \textendash{}
  Math\'ematiques}, 10:703--726, 2023.

\bibitem{cwikel_weak_1977}
M.~Cwikel.
\newblock Weak {{Type Estimates}} for {{Singular Values}} and the {{Number}} of
  {{Bound States}} of {{Schrodinger Operators}}.
\newblock {\em Annals of Mathematics}, 106(1):93--100, July 1977.

\bibitem{de_palma_quantum_2021-1}
G.~De~Palma and D.~Trevisan.
\newblock Quantum {{Optimal Transport}} with {{Quantum Channels}}.
\newblock {\em Annales Henri Poincar\'e}, 22(10):3199--3234, Oct. 2021.

\bibitem{deleporte_universality_2021}
A.~Deleporte and G.~Lambert.
\newblock Universality for {{Free Fermions}} and the {{Local Weyl Law}} for
  {{Semiclassical Schr\"odinger Operators}}.
\newblock {\em Journal of the European Mathematical Society}, to appear:1--79,
  Sept. 2021.

\bibitem{fournais_semi-classical_2018}
S.~Fournais, M.~Lewin, and J.~P. Solovej.
\newblock The {{Semi-classical Limit}} of {{Large Fermionic Systems}}.
\newblock {\em Calculus of Variations and Partial Differential Equations},
  57(4):1--105, Aug. 2018.

\bibitem{fournais_optimal_2020}
S.~Fournais and S.~Mikkelsen.
\newblock An optimal semiclassical bound on commutators of spectral projections
  with position and momentum operators.
\newblock {\em Letters in Mathematical Physics}, 110(12):3343--3373, Dec. 2020.

\bibitem{hundertmark_new_2000}
D.~Hundertmark, A.~Laptev, and T.~Weidl.
\newblock New bounds on the {{Lieb}}\textendash{{Thirring}} constants.
\newblock {\em Inventiones Mathematicae}, 140(3):693--704, June 2000.

\bibitem{lafleche_quantum_2022}
L.~Lafleche.
\newblock On {{Quantum Sobolev Inequalities}}.
\newblock {\em arXiv:2210.03013}, pages 1--24, Oct. 2022.

\bibitem{lafleche_quantum_2023}
L.~Lafleche.
\newblock Quantum {{Optimal Transport}} and {{Weak Topologies}}.
\newblock {\em arXiv:2306.12944}, pages 1--25, June 2023.

\bibitem{lafleche_strong_2023}
L.~Lafleche and C.~Saffirio.
\newblock Strong {{Semiclassical Limits}} from {{Hartree}} and
  {{Hartree}}\textendash{{Fock}} to {{Vlasov}}\textendash{{Poisson Equations}}.
\newblock {\em Analysis \& PDE}, 16(4):891--926, June 2023.

\bibitem{lewin_describing_2010}
M.~Lewin.
\newblock Describing lack of compactness in {{Sobolev}} spaces.
\newblock Jan. 2010.

\bibitem{lieb_inequalities_1976}
E.~H. Lieb and W.~E. Thirring.
\newblock Inequalities for the {{Moments}} of the {{Eigenvalues}} of the
  {{Schr\"odinger Hamiltonian}} and their {{Relation}} to {{Sobolev
  Inequalities}}.
\newblock {\em Studies in Mathematical Physics, Essays in Honor of Valentine
  Bargmann}, pages 269--303, 1976.

\bibitem{lions_concentration-compactness_1984-1}
P.-L. Lions.
\newblock The concentration-compactness principle in the calculus of
  variations. {{The}} locally compact case, part 1.
\newblock {\em Annales de l'Institut Henri Poincar\'e C, Analyse non
  lin\'eaire}, 1(2):109--145, 1984.

\bibitem{lions_concentration-compactness_1984}
P.-L. Lions.
\newblock The concentration-compactness principle in the calculus of
  variations. {{The}} locally compact case, part 2.
\newblock {\em Annales de l'Institut Henri Poincar\'e C, Analyse non
  lin\'eaire}, 1(4):223--283, 1984.

\bibitem{lions_sur_1993}
P.-L. Lions and T.~Paul.
\newblock Sur les mesures de {{Wigner}}.
\newblock {\em Revista Matem\'atica Iberoamericana}, 9(3):553--618, 1993.

\bibitem{martin_bound_1972}
A.~Martin.
\newblock Bound states in the strong coupling limit.
\newblock {\em Helvetica Physica Acta}, 45:140--148, July 1972.

\bibitem{porta_mean_2017}
M.~Porta, S.~Rademacher, C.~Saffirio, and B.~Schlein.
\newblock Mean {{Field Evolution}} of {{Fermions}} with {{Coulomb
  Interaction}}.
\newblock {\em Journal of Statistical Physics}, 166(6):1345--1364, Mar. 2017.

\bibitem{sickel_regularity_2021}
W.~Sickel.
\newblock On the {{Regularity}} of {{Characteristic Functions}}.
\newblock In {\em Anomalies in {{Partial Differential Equations}}}, volume~43
  of {\em Springer {{INdAM Series}}}, pages 395--441, {Cham}, 2021. {Springer
  International Publishing}.

\bibitem{simon_analysis_1976}
B.~Simon.
\newblock Analysis {{With Weak Trace Ideals}} and the {{Number}} of {{Bound
  States}} of {{Schr\"odinger Operators}}.
\newblock {\em Transactions of the American Mathematical Society},
  224(2):367--380, 1976.

\bibitem{tamura_asymptotic_1974}
H.~Tamura.
\newblock The asymptotic eigenvalue distribution for non-smooth elliptic
  operators.
\newblock {\em Proceedings of the Japan Academy, Series A, Mathematical
  Sciences}, 50(1):19--22, Jan. 1974.

\bibitem{tartar_introduction_2007}
L.~Tartar.
\newblock {\em An {{Introduction}} to {{Sobolev Spaces}} and {{Interpolation
  Spaces}}}, volume~3 of {\em Lecture {{Notes}} of the {{Unione Matematica
  Italiana}}}.
\newblock {Springer-Verlag Berlin Heidelberg}, {Berlin, Heidelberg}, 1 edition,
  2007.

\bibitem{triebel_interpolation_1978}
H.~Triebel.
\newblock {\em Interpolation {{Theory}}, {{Function Spaces}}, {{Differential
  Operators}}}.
\newblock Number~18 in North-{{Holland Mathematical Library}}. {Elsevier
  Science}, 1978.

\bibitem{triebel_fractals_2010}
H.~Triebel.
\newblock {\em Fractals and {{Spectra}}: {{Related}} to {{Fourier Analysis}}
  and {{Function Spaces}}}.
\newblock {Springer Science \& Business Media}, Oct. 2010.

\bibitem{werner_quantum_1984}
R.~Werner.
\newblock Quantum harmonic analysis on phase space.
\newblock {\em Journal of Mathematical Physics}, 25(5):1404--1411, May 1984.

\end{thebibliography}

\end{document}